\documentclass[12pt]{article}

\usepackage{url}
\usepackage{ifthen,amssymb}
\usepackage{cite}
\usepackage[cmex10]{amsmath} % Use the [cmex10] option to ensure complicance
\usepackage{amsfonts,mathtools}
\usepackage{enumitem}
\usepackage{amsthm}

\def\textiid{i.i.d.\@\xspace}
\newcommand\iid{\ifmmode\text{ i.i.d. } \else \textiid \fi}

\newcommand{\beqs}{\begin{equation*}}
\newcommand{\eeqs}{\end{equation*}}
\newcommand{\beq}{\begin{equation}}
\newcommand{\eeq}{\end{equation}}
\usepackage{xcolor}
\usepackage{appendix}
\usepackage[a4paper, total={6in, 8in},margin=2cm]{geometry}
\usepackage{hyperref}
\hypersetup{
    colorlinks=true,
    linkcolor=blue,
    filecolor=magenta,      
    urlcolor=cyan,
}

\DeclareMathOperator*{\argmin}{arg\,min}

\theoremstyle{plain}
\newtheorem{theorem}{Theorem}
\newtheorem{lemma}{Lemma}

\newtheorem{corollary}{Corollary}

\newtheorem{definition}{Definition}
\newtheorem{remark}{Remark}
\newtheorem{example}{Example}

\newcommand{\A}{{\mathsf{A}}}

\newcommand{\D}{{\mathsf{D}}}

\newcommand{\tw}{{\tilde{\textbf{w}}}}

\newcommand{\tW}{{\tilde{\textbf{W}}}}

\begin{document}

\title{
Learning under Distribution Mismatch and Model Misspecification
}

\author{Saeed Masiha\qquad Amin Gohari \qquad Mohammad Hossein Yassaee \\
 Mohammad Reza Aref}

\allowdisplaybreaks

\maketitle

\begin{abstract}
We study learning algorithms when there is a mismatch between the distributions of the training and test datasets of a learning algorithm. The effect of this mismatch on the generalization error and model misspecification are quantified. Moreover, we provide a connection between the generalization error and the rate-distortion theory, which allows one to utilize bounds from the rate-distortion theory to derive new bounds on the generalization error and vice versa. In particular, the rate-distortion based bound strictly improves over the earlier bound by Xu and Raginsky even when there is no mismatch. We also discuss how ``auxiliary loss functions" can be utilized to obtain upper bounds on the generalization error.

\end{abstract}

%%%%%%%%%%%%%%%%%%%%%%%%%%%%%%%%%%%%%%%%%%%%%%%%%%%%%%%%%%%%%%%%%

\section{Introduction}
In a learning algorithm, a \emph{distribution mismatch} occurs when the training dataset and the test dataset are not drawn from the same distribution. This mismatch might also occur if training data is corrupted or if the statistical distribution of the data changes from training to testing.  For example, suppose that (in the Covid era) a pharmaceutical company located in region $\mathrm{R}$ has developed a drug for Covid-19 (in statistical terms, the company has tuned the parameters of a process that describes how to mix different chemicals to make a drug). Clinical experiments show high effectiveness (say 95\%) of this treatment for the population that resides in region $\mathrm{R}$. There is an urgent need for the drug and the company lacks time to test the medicine on other populations with possibly different genetic  backgrounds (in statistical terms, with a different distribution from the distribution of the population in region $\mathrm{R}$). Hence it is required to have  some guarantee  on how the   effectiveness of treatment for the population $\mathrm{R}$ generalizes to other populations. As another example, in federated learning, a centralized model is trained based on chunks of training data originating from a number of clients, which may be mobile phones, other mobile devices, or sensors. While the training data may come from only a limited number of clients, statistical guarantees on the learning algorithm should be expressed in terms of testing on a population-averaged model of \emph{all} client distributions, which might be different from the training distribution. 

Distribution mismatch can manifest itself in different ways: consider a data scientist in a company who is given access to a training dataset and asked to make a recommendation about a  decision for the company. 
The training dataset is corrupted and its distribution slightly differs from that of the test data. The data scientist might run a learning algorithm $\A$ and utilize its output on the training data to make a recommendation. In the first part of this paper, we study the effect of distribution mismatch on the generalization error of algorithm $\A$. Next, assume that the company's manager impresses upon the data scientist the importance of the decision for the company and asks about his confidence level about his recommendation. To address this question, the data scientist needs to come up with a \emph{mathematical model} for the data and give guarantees based on that model. For instance, the data scientist might choose the parametric class of Gaussian distributions, partly based on the training data histograms (many methods to find a family of distributions for data samples are data-driven). Since the training data is corrupted, this process could lead to model misspecification. In the second part of this paper, we study how model misspecification affects theoretical guarantees of a learning algorithm.

\textbf{Generalization error under distribution mismatch:} Distribution mismatch is the subject of previous studies in transfer learning or domain adaptation \cite{mohri2019agnostic,mansour2009domain,wang2018theoretical,wu2020information,mansour2020theory}. An important goal common to domain adaptation and causal inference is to make accurate predictions when the distributions for the source (or training) domain(s) and target (or test) domain(s) differ \cite{magliacane2018domain}. Some interesting works on causal domain adaptation algorithms are discussed in \cite{chen2020domain,akbari2021recursive}. Distribution mismatch can be happened between labeled and unlabeled training data and test data in semi-supervised learning \cite{aminian2022information}. In \cite{aminian2022information}, they proposed a novel framework for self-training SSL algorithms that encompasses traditional SSL approaches such as the entropy minimization and the Pseudo-labeling approaches.  In the first part of this paper, we provide information-theoretic bounds on the generalization error under a distribution mismatch.  Designing algorithms with low generalization error is a key challenge in machine learning. It is known that under certain assumptions, the generalization error of a learning algorithm can be bounded from above in terms of the mutual information between the input and output of the algorithm \cite{russo2019much,xu2017information} (see also \cite{bu2020tightening,lopez2018generalization,wang2019information,hellstrom2020generalization,aminian2020jensen,esposito2020robust,issa2019strengthened,esposito2019generalization,jiao2017dependence,asadi2018chaining} for various generalizations and extensions using other measures of dependence). These works assume that the test data are drawn from the same distribution as the training data. Herein, we provide bounds on the generalization error of the learning algorithm assuming a bound on the KL divergence between the test and training distributions as well as a bound on the mutual information between the input and output of the learning algorithm. One of our bounds is based on (to the best of our knowledge) a novel connection between generalization error and the rate-distortion theory. When specialized to the case of no-mismatch, this bound strictly improves over the bound in
\cite{xu2017information} (see Corollary \ref{cor1} and Figure \ref{fig:screenshot001}). 

A question that we also address in this section is as follows: in case of having no mismatch between the training and test distributions, having more training data samples leads to increasingly better estimates of the unknown distribution of the data. On the other hand, in case of a mismatch, increasing the number of training samples can only provide more information about the training distribution. In the limit of the number of samples going to infinity, we will perfectly learn the training distribution but will still have a residual ambiguity about the test distribution: we will only know that the test distribution is at a certain KL distance from the training distribution. If we are in a regime where the error is dominated by this residual ambiguity in the test distribution,
the value of training samples gradually depreciates as we gather more samples. Subsequently, we might have insufficient incentive to gather more training samples. This shows that there is an ``optimal" number of samples associated with our problem. To the best of our knowledge, this question has not been addressed in the literature so far.
 We address the above question as follows:
In Corollary \ref{cor1}, we provide an upper bound of generalization error in terms of $\gamma+r/n$ where $\gamma$ is the KL-divergence between the training and test distributions, $r$ is the mutual information between the input and output of the learning algorithm and $n$ is the sample size. If $\gamma>0$ (i.e., when distribution mismatch exists), for large values of $n$ (or small values of $r$), the term $\gamma$ becomes the dominant term, and the effect of $r/n$ vanishes in the upper bound. This happens when $r/n$ is of the same order as $\gamma$. For a fixed sample size $n$ and $\gamma$, it suffices to work with algorithms that have input-output mutual information $r$ satisfying $r\approx n\gamma$.
In other words, since the training data is drawn from a different distribution than the test data, limited overfitting will not affect the generalization error.  Next, we give a  \emph{lower bound} on the generalization error in Corollary \ref{cor2} under distribution mismatch. Similar to the upper bound, this lower bound on the generalization error involves the summation of two terms. The first term is a constant (bounded from above by the KL-divergence between the training and test distribution, e.g. see \eqref{eq9911}) and another term (depending on the input-output mutual information of the algorithm) and vanishing in $n$. 
Finally, we also consider the performance of the \text{ERM} algorithm under distribution mismatch in Theorem \ref{th9}. We present an upper bound on excess risk. Increasing the number of samples does not make the upper bound vanish and we get a constant upper bound (due to distribution mismatch) when the number of samples tends to infinity.

\textbf{Model misspecification:} A learning algorithm has access to a training dataset that is drawn from an unknown distribution. This unknown distribution is commonly assumed to belong to a known family of distributions $\mathcal{P}$. A learning ``model" provides a description for the family $\mathcal{P}$, and a learning algorithm is required to have good performance when the data is drawn from any arbitrary distribution belonging to $\mathcal{P}$. We say that model misspecification occurs when the data distribution does not belong to $\mathcal{P}$. The amount of misspecification may be measured by the minimum KL-divergence from the true distribution to the family of distributions in class $\mathcal{P}$ \cite{wang2019variational}. Model misspecification is a key consideration in statistics \cite{masegosa2019learning,wang2019variational}. For instance, \cite{masegosa2019learning} shows that Bayesian methods are not optimal for learning predictive models unless the model class is perfectly specified. In the second part of the paper, we fix a uniformly stable learning algorithm $\A$ and assume a notion of sample complexity for the class $\mathcal{P}$. Then we bound the sample complexity under a distribution $\mu'\notin \mathcal{P}$ based on the minimum KL-divergence of $\mu'$ from the family $\mathcal{P}$. 

\textbf{Organization:}
The rest of this paper is organized as follows. The paper splits into two parts: section \ref{secDM} gives our results on generalization error while Section \ref{sec:Reliability} is dedicated to model misspecification. In Section \ref{sec:definition} we formally define learning with mismatched (training and test data) distributions. Section \ref{generalization} provides a connection between the rate-distortion theory and the generalization error, along with upper and lower bounds on the generalization error. The performance of the ERM algorithm on the training data when there is a distribution mismatch is also studied. Section \ref{sec:Reliability} studies model mismatch  for the class of uniformly-stable algorithms. Finally, Section \ref{further-ideas} discusses some ideas to improve the upper bounds on the generalization error given in Section \ref{generalization}.

\textbf{Notation and preliminaries:}
Random variables are shown in capital letters, whereas their realizations are shown in lowercase letters. We show sets with calligraphic font. For a random variable $X$ generated from a distribution $\mu$, we use $\mathbb{E}_{X\sim\mu}$ to denote the expectation taken over $X$ with distribution $\mu$ and $P_{X}$ means the distribution over $X$. We use $D(\mu\|\nu)$ and $D_\alpha(\mu\|\nu)=\frac{1}{\alpha-1}\log\int \left(\frac{d\mu}{d\nu}\right)^{\alpha}d\nu(x)$ to denote the KL divergence and the Renyi divergence of order $\alpha$ respectively. In particular, we have
$D_2(\mu\|\nu)=\log\left(1+\chi^{2}(\mu\|\nu)\right)$ where  $\chi^{2}$-divergence is defined as $\chi^{2}(\mu\|\nu)=\mathbb{E}_{\nu}\left(\frac{d\mu}{d\nu}-1\right)^{2}$. 
Given two random variables $X$ and $Y$, we use the shorthand $X\leq_{1-\delta} Y$ to denote
$\mathbb{P}[X\leq Y]\geq 1-\delta$. 
Observe that $X\leq_{1-\delta_1}Y$ and $Y\leq_{1-\delta_2}Z$ implies  $X\leq_{1-\delta_1-\delta_2}Z$ by the union bound.

 We write $g(n)\sim \Omega(f(n))$ when $g(n)\ge c\cdot f(n)$ for large enough $n$.
 
The concept of subgaussianity is defined as follows:
\begin{definition}\label{def2}
	The random variable $X$ is said to be sub-Gaussian with parameter $\sigma^2$ if $\forall s\in\mathbb{R}$
	\begin{align}
		\mathbb{E}[e^{s(X-\mathbb{E}[X])}]\le e^{\frac{\sigma^{2}s^{2}}{2}}.
	\end{align}
	Using the Chernoff's bound, we obtain,
	\begin{align}
		\mathbb{P}\left(|X-\mathbb{E}X|>t\right)\le e^{\frac{-t^{2}}{2\sigma^{2}}}.
	\end{align}
\end{definition}

The following lemma  relates the expectation of a measurable function over two different distributions:

\begin{lemma}[Donsker-Varadhan]\label{lemma1}
	Let $\mathcal{X}$ be a sample space and let $P$ be a distribution on $\mathcal{X}$. Let $Q$ be a distribution on $\mathcal{X}$ with the support which is a subset of the $P$ support. Then for any measurable function $\phi:\mathcal{X}\to \mathbb{R}$ with respect to $P$, we have
	\begin{align*}
		\ln\left(\mathbb{E}_{P}[e^{\phi(X)}]\right)\ge \mathbb{E}_{Q}[\phi(X)]- {D}(Q\|P).
	\end{align*}
\end{lemma}
\begin{lemma}\label{lmm2} [Coupling]
Given the marginal distributions $\mu$ and $\mu'$ on $\mathcal{Z}$, one can find a coupling $\pi(z,z')$ on $(z,z')\in\mathcal{Z}\times\mathcal{Z}$ such that $(Z,Z')\sim \pi$ satisfy $Z\sim \mu$, $Z'\sim \mu'$ and
$\mathbb{P}_{\pi}\left[Z\neq Z'\right]=\|\mu-\mu'\|_{TV}$ where $\|\mu-\mu'\|_{TV}$ is defined as
\[
\|\mu-\mu'\|_{TV}=\sup_{A\in \mathcal{Z}}\left[\mu(A)-\mu'(A)\right].
\]
\end{lemma}

%%%%%%%%%%%%%%%%%%%%%%%%%%%%%%%%%%%%%%%%%%%%%%%%%%%%%%%%%%%%%

%%%%%%%%%%%%%%%%%%%%%%%%%%%%%%%%%%%%%%%%%%%%%%%%%%%%%%%%%%%%%%
\section{Generalization error under distribution mismatch}
\label{secDM}
\subsection{Problem Definition}
\label{sec:definition}
Consider an instance space $\mathcal{Z}$, a hypothesis space $\mathcal{W}$ and a non-negative loss function $\ell:\mathcal{W}\times \mathcal{Z}\to \mathbb{R}^{+}$. Assume that the test and training samples are produced (in an i.i.d. fashion) from two unknown distributions $\mu$ and $\mu'$ on $\mathcal{Z}$ respectively. A training  dataset of size $n$ is shown by the $n$-tuple, $S'=(Z'_{1},Z'_{2},\cdots,Z'_{n})\in\mathcal{Z}^n$ of i.i.d. random elements according to an unknown distribution $\mu'$. 
A learning algorithm is characterized by a probabilistic mapping $\A(\cdot)$ (a Markov Kernel) that maps training data $S'$ to the  random variable $W'=\A(S')\in\mathcal{W}$ as the output hypothesis. The population risk of a hypothesis $w\in \mathcal{W}$ is computed on the test distribution $\mu$ as follows: 
\begin{align}
	L_{\mu}(w)\triangleq\mathbb{E}_{\mu}[\ell(w,Z)]=\int_{\mathcal{Z}}\ell(w,z)\mu(dz), \qquad\forall w\in\mathcal{W}.
\end{align}
The goal of learning is to ensure that under any data generating distribution $\mu$, the population risk of the output hypothesis $W'$ is small, either
in expectation or with high probability. Since $\mu$ and $\mu'$ are unknown, the learning algorithm cannot directly compute $L_{\mu}(w)$ for any $w\in \mathcal{W}$, but can compute the empirical risk of $w$ on the training dataset $S'$ as an approximation, which is defined as
\begin{align}
	L_{S'}(w)\triangleq\frac{1}{n}\sum_{i=1}^{n}\ell(w,Z'_{i}).
\end{align}
The true objective of the learning algorithm, $L_{\mu}(W')$, is unknown to the learning algorithm while the empirical risk $L_{S'}(W')$ is known.  The generalization gap is defined as the difference between these two quantities as  \cite{wang2018theoretical, wu2020information}
\begin{align}\label{new_notion}
	\mathrm{gen}_{\mu}(W',S')= L_{\mu}(W')-L_{S'}(W'),
\end{align}
where $W'=\A(S')$ is the output of the algorithm $\A$ on the input $S'\sim (\mu')^{\otimes n}$. In common algorithms such as empirical risk minimization (ERM) and gradient descent, $L_{S'}(W')$ is minimized \cite{shalev2010learnability, hardt2016train}. Therefore, to control $L_{\mu}(W')$ we need to bound $\mathrm{gen}_{\mu}(W',S')$ from above (in expectation or with high probability). Observe that $\mathrm{gen}_{\mu}(W',S')$, as defined in \eqref{new_notion}, is a random variable and a function of $(S',W')$. The generalization error is the expected value of $\mathrm{gen}_{\mu}(W',S')$:
\begin{align}\label{new_notion2}
	\mathrm{gen}\left(\mu,\mu',\A\right)= \mathbb{E}\left[L_{\mu}(W')-L_{S'}(W')\right].
\end{align}
When there is no-mismatch, \emph{i.e.,} $\mu=\mu'$, we  denote the generalization error  by $\mathrm{gen}\left(\mu, \A\right)$ for simplicity. 

\subsection{Upper bound on the generalization error}\label{generalization}

%For simplicity, let us denote the training set in this case by $S=(Z_1, \cdots, Z_n)\sim \mu^{\otimes n}$ instead of $S'$, and the output of the algorithm $\A$ by $W=\A(S)$.
The following upper bound on the generalization error is given in \cite{xu2017information} (see also \cite{russo2019much}):
\begin{theorem}[\cite{xu2017information}]\label{th1}  Assume that there is no distribution mismatch, \emph{i.e.,} $\mu'=\mu$.
	Suppose $\ell(w,Z)$ is $\sigma^{2}$-subgaussian under $Z\sim \mu$ for all $w\in\mathcal{W}$. Take an arbitrary algorithm $\A$ that runs on a training dataset $S'$. Then the generalization error is bounded as
	\[
	\mathrm{gen}(\mu,\A)\le\sqrt{\frac{2\sigma^{2}}{n}I\big(S';\A(S')\big)}.
	\]
\end{theorem}
%Since the true distribution $\mu$ is unknown, the upper bound $I\big(S';\A(S')\big)$ may not be known 

Let us write the \emph{sharpest} possible bound on the generalization error given an upper bound $r$ on $I(S';\A(S'))$:
\begin{align}
\D_1(r)&\triangleq \sup_{P_{W'|S'}:~I(W';S')\leq r}
\mathbb{E}\left[L_{\mu}(W')-L_{S'}(W')\right]\label{eqnRD1}
\end{align}
where the supremum in \eqref{eqnRD1} is over all Markov kernels $P_{W'|S'}$ with a bounded input/output mutual information and $S'\sim (\mu')^{\otimes n}$. %Observe that $\D_1(r)$ is the best bound possible if the only available information about the algorithm is $I(W;S)\leq r$. 
We claim that $\D_1(r)$ is related to the rate-distortion function. To see this, consider a rate-distortion problem where the input symbol space is $\mathcal{S}$, the reproduction space is $\mathcal{W}$ and the following distortion function between a symbol $w$ and an input symbol $s$ is used:\footnote{
While the literature commonly takes the reproduction space to be the same as the input symbol space, the rate-distortion theory does not formally require that.
}
 $$\Delta(w,s)=L_{s}(w)-L_{\mu}(w).$$
With this definition, from \eqref{eqnRD1}, we obtain
\begin{align}
-\D_1(r)&
=\inf_{P_{W'|S'}:~I(W';S')\leq r}
\mathbb{E}\left[\Delta(W',S')\right]\label{eqnRD2}
\end{align}
which is in the rate-distortion form. 

With $\D_1(r)$ defined as in \eqref{eqnRD1}, it follows that for any arbitrary algorithm $\A$ with $I\big(S';\A(S')\big)\leq r$ we have
\[
\mathrm{gen}\left(\mu,\mu',\A\right)\le \D_1(r).
\]
This upper bound does not require any subgaussianity assumption on the loss function. From this viewpoint, Theorem \ref{th1} is just a convenient and explicit lower bound on a rate-distortion function under an extra assumption on the loss function (for the no distribution mismatch case). We formalize this intuition in Theorem \ref{th3}.

Computing the upper bound $\D_1(r)$ is a convex optimization problem and there are efficient algorithms for computing it \cite{blahut1972computation}. However, computation of the bound can be practically difficult if the sample size $n$ is large. The following theorem provides a computable upper bound that requires running an optimization when the sample size is just one. 
\begin{theorem} For any arbitrary loss function $\ell(w,z)$, and algorithm $\A$ that runs on a training dataset $S'$ of size $n$, we have
\[\mathrm{gen}\left(\mu,\mu',\A\right)\le \D_2\left(\frac{I(S';\A(S'))}{n}\right)
\]
where
\begin{equation}
\D_2(r)\triangleq \max_{P_{\hat W|Z'}:~I(\hat W;Z')\leq r}
\mathbb{E}\left[L_{\mu}(\hat W)-\ell(\hat W,Z')\right]\label{eqnRD22}
\end{equation}
where $Z'\in\mathcal{Z}$ is distributed according to $\mu'$. Furthermore, to compute the maximum in \eqref{eqnRD22}, it suffices to compute the maximum over all conditional distributions $P_{\hat W|Z}$ for $\hat{W}\in \mathcal{W}$ such that the support of $\hat W$ can be chosen of size at most $|\mathcal{Z}|+1$.\label{Thm2}
\end{theorem}
The proof is given in the Section \ref{ProofThm2}.

%{\color{purple}
%The following proposition shows that   the bound in Theorem \ref{Thm2} is a strict improvement over the one in \cite{xu2017information} when there is no mismatch.
%\begin{proposition}\label{propNew}
 %   Consider the case of no distribution mismatch, \emph{i.e.,} $\mu=\mu'$. Suppose that $\ell(w,Z)$ is $\sigma^{2}$-subgaussian under $Z\sim \mu$ for all $w\in\mathcal{W}$. Then, $D_2(r)\leq \sqrt{2\sigma^2 r}$ for all $r$. Equality occurs if and only if (or only if) ???. PLEASE  ADD THE PROOF.
%\end{proposition}
%}

\begin{remark}
 Remember that $D_1(r)$ is
the \emph{sharpest} possible bound on the generalization error given $I(S';\A(S'))\leq r$.
Thus, 
the upper bound 
$\D_{2}(r/n)$
of  Theorem \ref{Thm2} on the generalization error  is worse than the bound based on $\D_1(r)$, \emph{i.e.,} $\D_1(r)\leq \D_2(r/n)$. However, the bound  $\D_2(r/n)$ is easier to compute than $\D_1(r)$ because the optimization problem in \eqref{eqnRD22} is for a single symbol $Z'$ whereas the  optimization problem in \eqref{eqnRD1} is for a sequence $S'$ of $n$ symbols. Even though $\D_{2}(\cdot)$ is a rate-distortion function and does not admit an explicit closed-form expression in general, the Balhut-Arimoto algorithm can be used to compute it \cite{blahut1972computation} even when the cardinality of instance space $\mathcal{Z}$ is  infinite (see also \cite{blahut1972computation2}).\footnote{
Rate-distortion theory for continuous or abstract alphabets is discussed at length in the literature, e.g. see 
\cite{csiszar1974extremum, rose1994mapping}. See also  \cite{berger1998lossy} for a survey. 
} 

Later in Section \ref{further-ideas}, we show that ``auxiliary loss functions" can be utilized to tighten the gap between $\D_{1}(r)$ and $\D_{2}(r/n)$.

\end{remark}
While $\D_2(r)$ is easier to compute than $\D_1(r)$  and does not require any subguassianity assumption on the loss function, the bound in Theorem \ref{th1} is in a very explicit form.
Moreover, the bound in Theorem \ref{th1} (for the case of no-mismatch) depends only on mutual information $I\big(S';\A(S')\big)$ while the bound in Theorem \ref{Thm2}  depends on  $\mu$, $\mu'$ and $I\big(S';\A(S')\big)$. However, one can obtain a bound from Theorem \ref{Thm2} that does not depend on  $\mu$, $\mu'$ by maximizing the bound in Theorem \ref{Thm2} over all distributions $\mu$ and $\mu'$ such that  $D(\mu'\|\mu)\leq \gamma$ for some $\gamma>0$. We show that even after this maximization, the bound in Theorem \ref{Thm2} is still an improvement over Theorem \ref{th1}. To show this, we need to prove that the bound in Theorem \ref{Thm2} is always less than or equal to the bound in Theorem \ref{th1} for any arbitrary $\mu$ and $\mu'$ satisfying $D(\mu'\|\mu)\leq \gamma$ and the subgaussianity assumption on the loss function. Below, we give a general result for the rate-distortion function and deduce the relation between the bounds in
Theorem \ref{th1} and Theorem \ref{Thm2} as a corollary to it.

\begin{theorem}\label{th3}
    Consider a generic rate-distortion problem for $X\sim \zeta$ and a distortion function $d(x,\hat{x})\in\mathbb{R}$. Let 
$\phi(\cdot)$ be a function defined on $(-b,0]$ for some $b\in(0,\infty]$ as follows:
    \[
   \phi(\lambda)=\sup_{\hat{x}} \log\mathbb{E}_{\eta}\left[e^{\lambda d(X,\hat{x})}\right],
    \]
   for some distribution $\eta$ on $\mathcal{X}$ (possibly different from $\zeta$). Then,
     \begin{align}
     \inf_{P_{\hat{X}|X}:\,I(\hat{X};X)\le r}\mathbb{E}\left[d(X,\hat{X})\right]\ge \sup_{-b<\lambda<0}\left\lbrace\frac{1}{\lambda}\left[r+D(\zeta_X\|\eta_X)\right]+\frac{1}{\lambda}\phi(\lambda)\right\rbrace.
    \end{align}
\end{theorem}
Proof of Theorem \ref{th3} is in Section \ref{appenB}.

We apply the above theorem to obtain an upper bound on the bound given in 
Theorem \ref{Thm2} as follows: let $X=Z'\sim\mu'$, $\hat{X}=\hat{W}$ and $d(z',\hat w)=-	\left[L_{\mu}(\hat w)-\ell(\hat w,z')\right]$.

\begin{corollary}
	\label{cor1} 
	Suppose that $\ell(w,Z)$ is $\sigma^{2}$-subgaussian  for every $w\in\mathcal{W}$ under the distribution $\mu$ on $Z$. Take an arbitrary algorithm $\A$ that runs on a training dataset $S'$. Then when $I\big(S';\A(S')\big)\le r$ and $D(\mu'\|\mu)\leq \gamma$ for some $r,\gamma\geq 0$, then
	\begin{align}\label{rd&Mi}
	 \D_2(r)
		\le	\sqrt{2\sigma^{2}\gamma+\frac{2\sigma^{2}}{n}r}.
	\end{align}
\end{corollary}
\begin{remark}
Under the assumptions of Corollary \ref{cor1}, we deduce that 
\begin{align}
	\mathrm{gen}\left(\mu,\mu',\A\right)
		&\le
	\sqrt{2\sigma^{2}\gamma+\frac{2\sigma^{2}}{n}r}.
	\end{align}
	This  generalizes the bound in  Theorem \ref{th1} to the case of having mismatch.
\end{remark}

\begin{figure}
	   \centering
	\includegraphics[scale=1,width=0.8\linewidth]{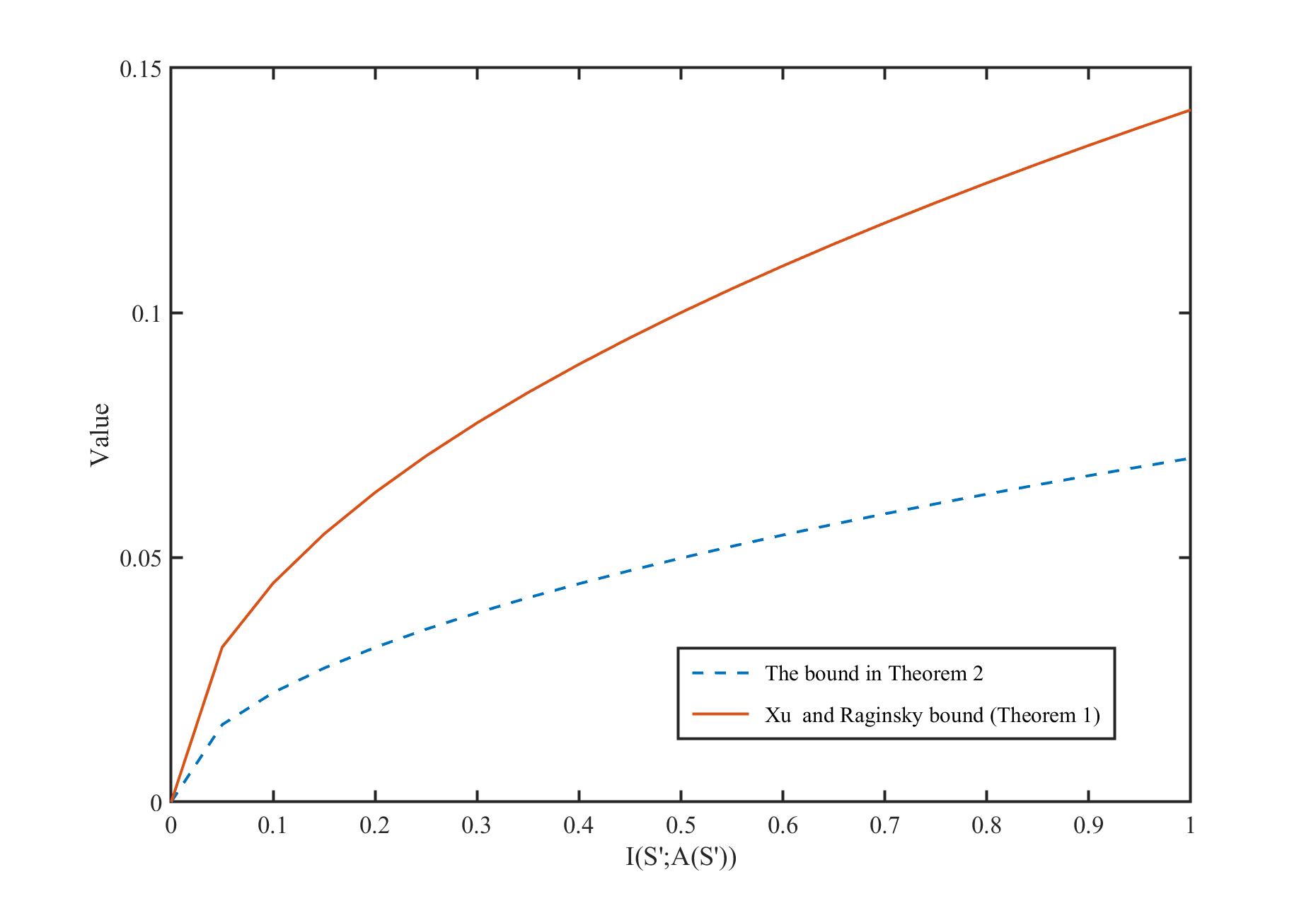}
	\caption{The bound in Theorem \ref{th1} versus the maximum of the upper bound in Theorem \ref{Thm2} over all distributions $\mu$, assuming no distribution mismatch, $\mathcal{W}=\mathcal{Z}=\{0,1\}$ and  $\ell(w,z)=w\cdot z$.}
	\label{fig:screenshot001}
\end{figure}
\begin{figure}
    \centering
    \includegraphics[scale=1,width=0.8\linewidth]{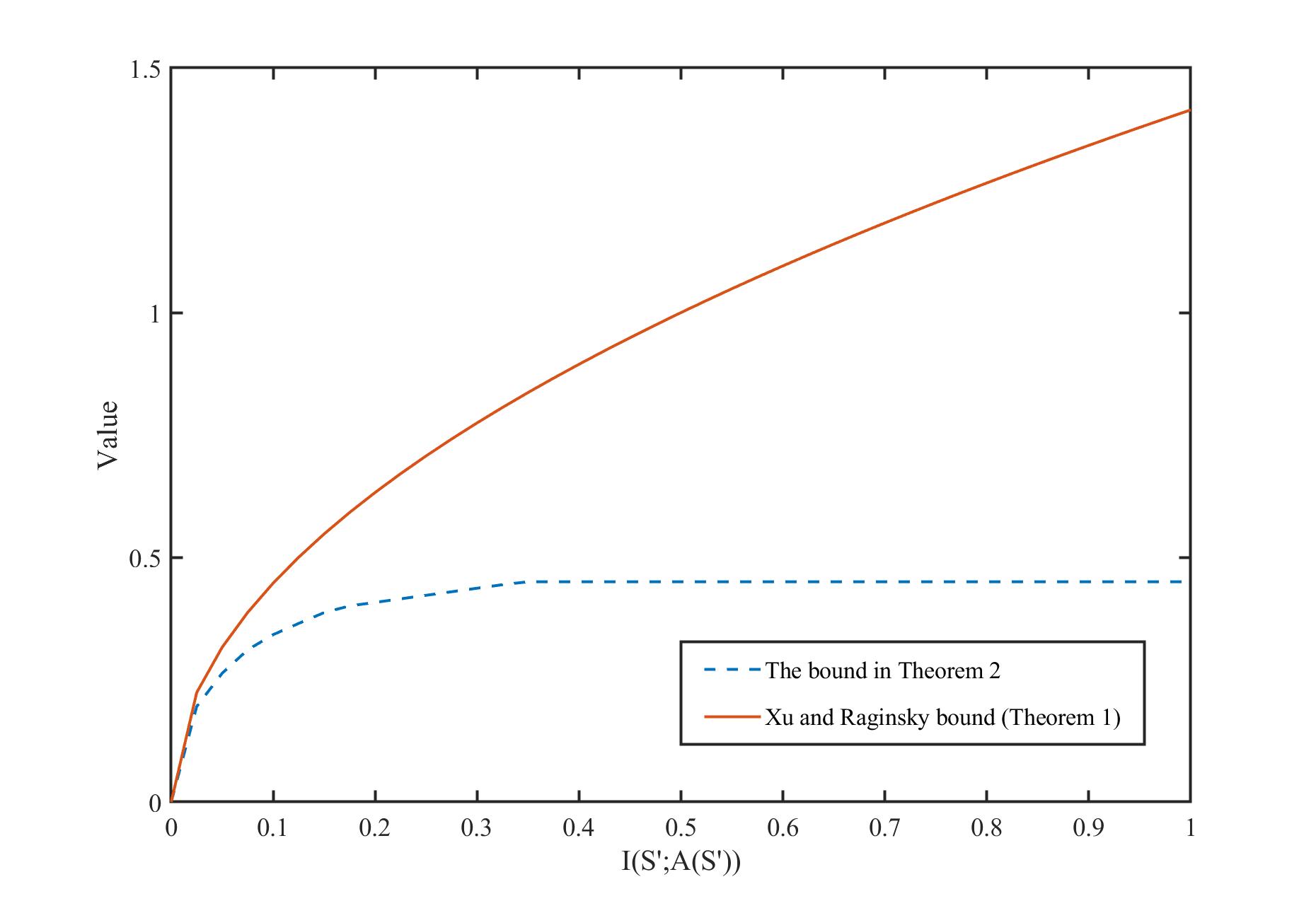}
    \caption{The bound in Theorem \ref{th1} versus the maximum of the upper bound in Theorem \ref{Thm2} over all distributions $\mu$, assuming no distribution mismatch, $\mathcal{W}=[0,1],\,\mathcal{Z}=\{0,1\}$ and  $\ell(w,z)=|w-z|$.}
    \label{fig:screenshot002}
\end{figure}
\begin{example}\label{ex1}
Let $\mathcal{W}=\mathcal{Z}=\lbrace0,1\rbrace$ and consider a learning problem on a data set $S'$ with the size $n=1$ with loss function $\ell(w,z)=w\cdot z$. Figure \ref{fig:screenshot001} depicts the bound  in Theorem \ref{th1} versus the maximum of the bound in Theorem \ref{Thm2}  over all distributions $\mu$ on $\lbrace0,1\rbrace$ for the case of no-mismatch for a particular loss function. Note that the distortion function itself depends on the choice of $\mu$ and this makes it difficult to find a closed form expression for the maximum of the bound in Theorem \ref{Thm2} over all distributions $\mu$. 
\end{example}
\begin{example}\label{ex2}
Let $\mathcal{W}=[0,1],\,\mathcal{Z}=\lbrace0,1\rbrace$ and consider a learning problem on a data set $S'$ with the size $n=1$ with loss function $\ell(w,z)=|w-z|$. Figure \ref{fig:screenshot002} depicts the bound  in Theorem \ref{th1} versus the maximum of the bound in Theorem \ref{Thm2}  over all distributions $\mu$ on $\lbrace0,1\rbrace$ for the case of no-mismatch for a particular loss function.
\end{example}

\iffalse
{\color{purple}
We wish to maximize
$$\max_{p(w,z):I(W;Z)\leq r}\sum_{w,z}\ell(w,z)\left[p(w)p(z)-p(w,z)\right]$$
or
$$\max_{p(w,z)}\sum_{w,z}\ell(w,z)\left[p(w)p(z)-p(w,z)\right]-\lambda I(W;Z)$$
where maximum is over the probability simplex where $\sum_{w,z}p(w,z)=1$. The derivative of the expression with respect to $p(w,z)$ for some fixed $(w,z)$ is
\[-\ell(w,z)+\ell(w)+\ell(z)-\lambda\log\frac{p(w,z)}{p(w)p(z)}\]
where
$\ell(w)=\sum_{z}p(z)\ell(w,z)$
$\ell(z)=\sum_{w}p(w)\ell(w,z)$. Since we fixed $\sum_{w,z}p(w,z)=1$, we deduce that there is a constant $K$ such that
\[K=-\ell(w,z)+\ell(w)+\ell(z)-\lambda\log\frac{p(w,z)}{p(w)p(z)}, \qquad\text{ if } p(w,z)>0.\]
Thus, letting $\beta=1/\lambda$, we get that whenever $p(w,z)>0$ we have
\[p(w,z)=p(w)p(z)e^{-\beta[\ell(w,z)+\ell(w)+ \ell(z)- K]}
\]
or
\[p(w,z)=A(w)B(z)e^{-\beta\ell(w,z)}\]
for some functions $A(\cdot)$ and $B(\cdot)$.
}
\fi

\subsubsection{An improved upper bound} In \cite{bu2020tightening}, a strengthened version of Theorem \ref{th1} is given as follows:
\begin{theorem}[\cite{bu2020tightening}]\label{thm5}
    Suppose that the  loss function $\ell(w,Z)$ is $\sigma^{2}$-subgaussian under the distribution $\mu$ on $Z$ for any $w\in \mathcal{W}$. For $\mu'=\mu$, we have:
    \begin{align}
    \mathrm{gen}(\mu,\mu',\A)\le\frac{1}{n}\sum_{i=1}^{n}\sqrt{{2\sigma^{2}}I\big(Z_{i}';\A(S')\big)}.
    \end{align}
\end{theorem}

The following variant of Theorem \ref{thm5} holds for the case with distribution mismatch:
\begin{theorem}
    \label{Thm2b2}
    For any arbitrary loss function $\ell(w,z)$, and algorithm $\A$ that runs on a training dataset $S'=(Z'_1, Z'_2, \cdots, Z'_n)$ of size $n$, we have
\[\mathrm{gen}\left(\mu,\mu',\A\right)\le \frac1n\sum_{i=1}^n \D_2\left(I(Z'_i;\A(S'))\right)
\]
where $\D_2(r)$ is given in \eqref{eqnRD22}. 
Moreover, if the loss function $\ell(w,Z)$ under $\mu$ is $\sigma^{2}$-subgaussian  for all $w\in \mathcal{W}$, we further have
\[\mathrm{gen}\left(\mu,\mu',\A\right)\le \frac1n\sum_{i=1}^n \D_2\left(I(Z'_i;\A(S'))\right)\le \frac{1}{n}\sum_{i=1}^{n}\sqrt{{2\sigma^{2}}\left[I\big(Z_{i}';\A(S')\big)+D(\mu'\|\mu)\right]}.
\]
\end{theorem}

Proof of Theorem \ref{Thm2b2} is in Section \ref{appenThm2b2}.

%%%%%%%%%%%%%%%%%

%%%%%%%%%%%%%%%%%
\subsection{Lower bound on the generalization error}
Next, we consider \emph{lower bounds} on the generalization error. Similar to \eqref{eqnRD1}, the following lower bound on the generalization error given an upper bound $r$ on $I(S';\A(S'))$ can be written:
\begin{align} \inf_{P_{W'|S'}:~I(W';S')\leq r}
\mathbb{E}\left[L_{\mu}(W')-L_{S'}(W')\right]\label{eqnRD1lower}
\end{align}
where the infimum in \eqref{eqnRD1lower} is over all Markov kernels $P_{W'|S'}$ with a bounded input/output mutual information and $S'\sim (\mu')^{\otimes n}$. However, the bound in this form may not be useful. To see this, assume that $\mu=\mu'$. One possible choice for $W'$ in \eqref{eqnRD1lower} is a constant random variable. For this choice, $I(W';S')=0\leq r$ and the bound in \eqref{eqnRD1lower} vanishes. It follows that $\D_3(r)\leq 0$. However, we are interested in a lower bound on the generalization error  in terms of the population risk. In order to prevent $W'$ from being a constant random variable, we attempt to find a lower bound on the generalization error in terms of both $I(S';\A(S'))$ and an assumption about the marginal distribution of the output of the algorithm $\A(S')$. In particular, we assume that $I(S';\A(S'))\leq r$ and $\A(S')\sim p_{W'}\in \mathcal{M}$ for a family $\mathcal{M}$ of distributions on $\mathcal{W}$. 

We aim to find a lower bound on $\mathrm{gen}\left(\mu,\mu',\A\right)$ that depends on both $r$ and $\mathcal M$. The sharpest such bound is 
\begin{align}
\D_3\big(r,\mathcal{M}\big)= \inf_{P_{W'}\in\mathcal{M}}~~\inf_{\substack{P_{W',S'}\in U(P_{W'},P_{S'}):\\I(W';S')\leq r}}
\mathbb{E}\left[L_{\mu}(W')-L_{S'}(W')\right]\label{eqnRD1lower2}
\end{align}
where $U(P_{W'},P_{S'})$ is the set of all couplings of two marginal distribution $P_{W'}$ and $P_{S'}$ on $\mathcal{W}\times \mathcal{S}$. 

 When $r=0$, the set $U(P_{W'},P_{S'})$ includes only the product distribution $P_{W'}P_{S'}$ and $\D_3\big(0,\mathcal{M}\big)$ can be computed explicitly. 
 The following theorem gives an explicit lower bound on the generalization error when $r>0$:
 \begin{theorem}\label{thm4}Let $\psi(\lambda)$ be a function satisfying
 \begin{align*}
     \psi(\lambda)\geq \sup_{\nu\in \mathcal{M}}\mathbb{E}_{\nu}\left[e^{\lambda\left[\ell(W,z)-\mathbb{E}_{\nu}[\ell(W,z)]\right]}\right], \quad\forall z\in \mathcal{Z}.
\end{align*}
Then, we have:
     \[\D_3\big(0,\mathcal{M}\big)\geq \D_3\big(r,\mathcal{M}\big)
    \geq
	\D_3\big(0,\mathcal{M}\big)-\inf_{\lambda\ge0}\left[\lambda r-\lambda(\psi({1}/{n\lambda})^{n}-1)\right].
	\]
 \end{theorem}
\begin{corollary}\label{cor2}
 Suppose that $\ell(W',z)$ is $\alpha^{2}$-subgaussian under any $P_{W'}\in\mathcal{M}$ for all $z\in \mathcal{Z}$. Considering the special choice  of $\lambda={\alpha}/{\sqrt{2nr}}$, we deduce 
\[\D_3\big(0,P_{W'}\big)\geq \D_3\big(r,P_{W'}\big)
    \geq
	\D_3\big(0,P_{W'}\big)-\frac{1}{\sqrt{n}}\left[\frac{\alpha\sqrt{r}}{\sqrt{2}}+\frac{\alpha}{\sqrt{2r}} (e^{r}-1)\right].
	\]
	Therefore,
	\[\mathrm{gen}\left(\mu,\mu',\A\right)\ge
	\D_3\big(0,P_{W'}\big)-\frac{1}{\sqrt{n}}\left[\frac{\alpha\sqrt{I(S';\A(S'))}}{\sqrt{2}}+\frac{\alpha}{\sqrt{2I(S';\A(S'))}} (e^{I(S';\A(S'))}-1)\right].
	\]
\end{corollary}
Proof of the above theorem can be found in Section \ref{Sec:ProofThm4}.
The following theorem gives another lower bound on the generalization error which can be compared with the upper bound in Theorem \ref{Thm2}:
\begin{theorem} For any arbitrary loss function $\ell(w,z)$, and algorithm $\A$ that runs on a training dataset $S'$ of size $n$ and induces a marginal distribution on $\A(S')$ in $\mathcal{M}$, we have
\[\mathrm{gen}\left(\mu,\mu',\A\right)\ge \D_4\left(\frac{I(S';\A(S')}{n},\mathcal{M}\right)
\]
where
\begin{equation}
\D_4(r,\mathcal{M})\triangleq \inf_{P_{W'}\in\mathcal{M}}~~\min_{P_{\hat W,Z'}\in U(P_{W'},\mu'):~I(\hat W;Z')\leq r}
\mathbb{E}\left[L_{\mu}(\hat W)-\ell(\hat W,Z')\right].\label{eqnRD22b}
\end{equation}
\label{Thm2b}
\end{theorem}

The proof is given in the Section \ref{appenDE}.

\textbf{High probability guarantees:}
Just as the \emph{excess distortion} probability of a rate-distortion code has been subject of many studies in information theory (see \cite{matsuta2015non, marton1974error} for two examples), a number of ``high probability" upper bounds on the generalization gap are also reported in the literature. Here the problem is to find an upper bound on 
\begin{align*}
    \mathbb{P}[\mathrm{gen}_{\mu}(W',S')\geq \eta]
&=\mathbb{P}[
L_{\mu}(W')-L_{S'}(W')\geq \eta]
\end{align*}
for some given $\eta$.

The following bound is a generalization of a bound in  \cite{esposito2019generalization} to include distribution mismatch. Our method for deriving this inequality is different from the one used in  \cite{esposito2019generalization}, and similar to the one used in \cite{hellstrom2020generalization}.

\begin{theorem}\label{lhighprobability_on_gen}
Take some algorithm $\A$ that runs on a training dataset $S'$ and produces an output hypothesis $W'=\A(S')$. Let $\ell(w,Z)$ be a loss function which is $\sigma^{2}$-subgaussian under the  distribution $\mu$ on $Z$ for all $w\in\mathcal{W}$. Then, we have    
\begin{align}
    \mathbb{P}[|\mathrm{gen}_{\mu}(W',S')|\geq \eta]\leq 2\exp\left(-\frac{n\left(\frac{\eta^{2}}{2}-\sigma^{2}D_{2}(\mu'\|\mu)\right)-\sigma^{2}D_{2}(P_{W'S'}\|P_{W'}P_{S'})}{3\sigma^{2}}\right).
\end{align}
\end{theorem}
Proof of the above theorem is given in Section \ref{Sec:highProb}.

\textbf{Performance of the ERM algorithm:} 
As an application of Theorem \ref{lhighprobability_on_gen}, let us consider the \text{ERM} algorithm which is defined as follows:
\begin{align}\label{eq125}
	W_{ERM}(S')=\argmin_{w\in \mathcal{W}}L_{S'}(w).
\end{align}
Then, we claim the following upper bound on excess risk of the \text{ERM} algorithm: 

\begin{theorem}\label{th9}
Let $\ell(w,Z)$ be  $\sigma^{2}$-subgaussian under the distribution $\mu$ on $Z$ for every $w$. Consider the \text{ERM} learning algorithm $\A_{ERM}$ as defined in \eqref{eq125}. Then, with probability of at least $1-\delta$,
    \begin{align}
        {L}_{\mu}(\A_{ERM}(S'))\le&\min_{w\in \mathcal{W}}{L}_{\mu}(w)+\sqrt{2\sigma^{2}D_2(\mu'\|\mu)+2\sigma^{2}\frac{\log\left(\frac{4}{\delta}\right)}{n}}\nonumber\\
        &+\sqrt{2\sigma^{2}D_2(\mu'\|\mu)+\frac{2\sigma^{2}\left[D_{2}(P_{W'S'}\|P_{W'}P_{S'})+3\log\left(\frac{4}{\delta}\right)\right]}{n}}.
    \end{align}
    \label{thm8}
\end{theorem}
Proof of the above theorem can be found in Section \ref{sec:proof8}.

%%%%%%%%%%%%%%%%%%%%%%%%%%%%%%%%%%%%%%%%%%%%%%%%%%%%%%%%%%%%%%%%%

%%%%%%%%%%%%%%%%%%%%%%%%%%%%%%%%%%%%%%%%%%%%%%%%%%%%%%%%%%%%%%
\section{Learning under Model Misspecification }\label{sec:Reliability}
Take an algorithm $\A$ along with a sample-complexity guarantee for a family of distributions $\mathcal{P}$, \emph{i.e.,} the model has specified the class $\mathcal{P}$. Given $\delta,\epsilon>0$, sample complexity is defined as
\begin{align*}
	n(\A, \mathcal{P}, \epsilon, \delta)=\min\left\{N: \forall n>N, ~\sup_{\mu\in\mathcal{P}}\mathbb{P}\left[L_{\mu}(\A(S))-\min_{w\in \mathcal{W}}L_{\mu}(w)\le \epsilon\right]\ge1-\delta\right\}
\end{align*} 
where $S=(Z_1,Z_2,\cdots,Z_n)\sim \mu^{\otimes n}$ is the training data. 
We would like to find the increase in sample-complexity if we expand the set $\mathcal{P}$ to 
$$\mathcal{P}_{\gamma}=\{\mu': \inf_{\mu\in \mathcal{P}}D(\mu'\|\mu)\le \gamma\}.$$
The set $\mathcal{P}_{\gamma}$ relates to model misspecification when it is llimited in KL divergence of at most $\gamma$. We utilize the following alternative equivalent definition of sample-complexity:
\begin{align*}
	\epsilon(\A, \mathcal{P}, n, \delta)=\inf\left\{x\in \mathbb{R}:~\sup_{\mu\in\mathcal{P}}\mathbb{P}\left[L_{\mu}(\A(S))-\min_{w\in \mathcal{W}}L_{\mu}(w)\le x\right]\ge1-\delta\right\}.
\end{align*} 
We restrict to uniformly-stable algorithms. In general terms, a learning algorithm is said to be stable if a small change of the input to the algorithm does not change the output of the algorithm much. Examples of stability definitions include uniform stability defined by Bousquet and Elisseeff \cite{bousquet2002stability}. 
The definition of stability that we adopt in this paper is as follows:
\begin{definition}\label{def0}
Given non-negative real numbers $\beta_i(n)$ we say that the $\A$ is called uniformly-stable if for any $s_1=(z_{11}, z_{12}, \cdots, z_{1n}), s_2=(z_{21}, z_{22}, \cdots, z_{2n})\in\mathcal{Z}^n$,  the following inequality holds (almost surely):
	\begin{align}
		|\ell(\A(s_1),z)-\ell(\A(s_2),z)|\le \sum_{i=1}^{n}\beta_{i}(n)\boldsymbol{1}[z_{1i}\neq z_{2i}],\quad\forall z\in\mathcal{Z}.
	\end{align}
\end{definition}

%%%%%%%%%%%%%%%%%%%%%%%%%%%%%%%%%%%%%%%%%%%%%%%%%%%%

	\begin{theorem}
		\label{thm:sample_mix_comp}
		Let $\ell(w,Z)$ be $\sigma^{2}$-subgaussian over $Z$ with distribution $\mu\in \mathcal{P}$ for every $w$. Then, for every $n,\gamma>0$ and $ \delta\in[0,1]$, we have
		\begin{align}\label{eq100}
		    	\epsilon(\A, \mathcal{P}_\gamma, n, {\delta})\le \epsilon(\A,  \mathcal{P}, n, \delta)+\sum_{i=1}^{n}\beta_{i}+2\sqrt{2\sigma^{2}\gamma},
		\end{align}
and
		\begin{align}\label{eq101}
		\epsilon(\A, \mathcal{P}_\gamma, n, {\delta})\le \epsilon(\A,  \mathcal{P}, n, \delta/2)+f(\delta),
		\end{align}
		where the function $f$ is defined as
		\begin{align}
	f(\delta)&\triangleq\sqrt{2\sigma^{2}\gamma}+\sqrt{2\sigma^{2}\left[\log(2/\delta)+\gamma\right]}\nonumber+\frac{1}{g(\delta/2)}\sum_{i=1}^{n}\ln\left(1+\frac{1}{2}(\exp(g(\delta/2)\beta_{i})-1)\sqrt{\gamma}\right)
,\\
g(\delta)&\triangleq\sqrt{\frac{2\left[\log(1/\delta)+\gamma\right]}{\sigma^{2}}}.
		\end{align}
		
	\end{theorem}

Proof of Theorem \ref{thm:sample_mix_comp} is given in Section \ref{appenDD}.

\begin{remark}

	Equation \eqref{eq101} yields the following inequality on the sample complexities:
	\[
	n(\A,\mathcal{P}_{\gamma},\epsilon,\delta)\le n(\A,\mathcal{P},\epsilon-f(\delta,\gamma),\delta/2)
	\]
	for every $\epsilon\geq f(\delta,\gamma)$.

	\end{remark}

\begin{remark}
While the upper bound \eqref{eq100} is in terms of $\sum_{i=1}^{n}\beta_{i}$, the upper bound \eqref{eq101} is approximately in terms of $\sqrt{\gamma}\sum_{i=1}^{n}\beta_{i}$. Numerical simulations suggest that the bound \eqref{eq101} is better than the bound  \eqref{eq100} if $\beta_{i}\sim \Omega(\frac{1}{\sqrt{n}})$. More specifically, if we can choose $\gamma=o(1/n)$ in order to have non-divergent upper bound when $\beta_{i}\sim \Omega(\frac{1}{\sqrt{n}})$ i.e. there is a trade-off between $\gamma$ and $\beta_{i}$'s in order to have gaurantee on performance of a machine learning algorithm on misspecified model. The regime
$\beta_{i}\sim \Omega(\frac{1}{\sqrt{n}})$ could be of importance, e.g., see  \cite{shalev2010learnability,hardt2016train}.
\end{remark}

%%%%%%%%
%%%%%%%%
\section{Further ideas to improve the rate-distortion upper bound}
\label{further-ideas}
While $\D_1(r)$ (as defined in \eqref{eqnRD1}) is the sharpest 
possible bound on the generalization error given an upper bound $r$ on $I(S';\A(S'))$, the single-letter bound 
$\D_2(r/n)$ in Theorem \ref{Thm2} is not. In fact, the following relaxation is used in  the proof of Theorem \ref{Thm2}: instead of producing one output hypothesis $W$ for the entire sequence $S'=(Z'_1, Z'_2, \cdots, Z'_n)$, we produce $n$ output hypothesis $\tilde{W}_1, \tilde{W}_2, \cdots, \tilde{W}_n$. To tighten the gap between  $\D_1(r)$ and $\D_2(r/n)$, one needs to answer the following question: given a joint distribution $(W,Z'_1, Z'_2, \cdots, Z'_n)$, what are the set of marginal distributions on $(W,Z'_i)$? For instance if $W$ is a binary random variable and $(Z'_1, Z'_2, \cdots, Z'_n)$ are i.i.d., $W$ cannot have high dependence with all of the $Z'_i$'s.\footnote{In particular, using mutual information as the measure of dependence we have the following: for a binary $W$ and a sequence $(Z'_1, \cdots, Z'_n)$ of independent random variables, we have $1\geq I(W;Z'_1,Z'_2,\cdots,Z'_n)\geq \sum_{i}I(W;Z'_i)$. See \eqref{eqnRR12} for a proof. Thus, sum of  correlations between $W$ and $Z'_i$ is no more than one bit. }

Motivated by the above question, in the rest of this section we present a general idea which may be used on its own, or in conjunction with the  ideas in the previous section to improve the upper bound given in Theorem \ref{Thm2}. 
Let $\tilde\ell(w,z)$ be an ``auxiliary" loss function; an arbitrary loss function of our choice which can be different from the original loss function $\ell(w,z)$. We show that the average risk of the ERM algorithm on the auxiliary loss function $\tilde{\ell}$ can be used to bound the generalization error of a different algorithm $\A$, which runs on the same training data as the ERM algorithm, but with the original loss function $\ell(w,z)$. Let
$$\mathsf{ERM}(z'_1, \cdots, z'_n)=\min_{w}\sum_{i=1}^n\frac1n\tilde\ell(w,z'_i)$$
be the risk of the ERM algorithm given a training sequence $s'=(z'_1, z'_2, \cdots, z'_n)$ according to $\tilde\ell$. Let
$$v_n=\mathbb{E}_{S'\sim (\rho')^{\otimes n}}\mathsf{ERM}(Z'_1, \cdots, Z'_n)$$
be the average risk of the ERM algorithm. Let us, for now, assume that $v_n$ is known to us.

Take an arbitrary algorithm $\A$. Let $W'=\A(S')$ Then, the risk of $\A$ with respect to $\tilde{\ell}$ is greater than or equal the risk of the ERM algorithm, \emph{i.e.,}
\begin{align}\mathbb{E}\left[\sum_{i=1}^n\frac1n\tilde\ell(W',Z'_i)\right]\geq v_n\label{eqnArg1}.\end{align}
Let $Q$ be a random variable, independent of all previously defined variables, and uniform on the set $\{1,2,\cdots, n\}$. Set $\tilde{Z}=Z'_{Q}$. Observe that $\tilde{Z}\sim \mu'$ because $Z'_i\sim \mu'$ for all $i$ and $Q$ is independent of $(Z'_1,\cdots, Z'_n)$. 
Using this definition for $\tilde{Z}$, the risk of $\A$ with respect to the loss $\tilde\ell$ equals
\begin{align}\mathbb{E}[\tilde\ell(W',\tilde{Z})]=\mathbb{E}\left[\sum_{i=1}^n\frac1n\tilde\ell(W',Z'_i)\right]\label{eqnArg2}\end{align}
and the generalization error with respect to the loss $\ell$ can be characterized as
\begin{align}	\mathrm{gen}\left(\mu,\mu',\A\right)=\frac1n\sum_{i=1}^n\mathbb{E}\left[L_{\mu}(W')-\ell(W',Z'_i)\right]= \mathbb{E}\left[L_{\mu}(W')-\ell(W',\tilde{Z})\right].\label{eqnArg3}\end{align}
From \eqref{eqnArg1}, \eqref{eqnArg2} and \eqref{eqnArg3} we obtain the following upper bound on the generalization error of $\ell$:
\begin{align}\mathrm{gen}\left(\mu,\mu',\A\right)\le  \max_{P_{\hat W|\tilde{Z}}:~\mathbb{E}[\tilde\ell(\hat W,\tilde{Z})]\geq v_n}
\mathbb{E}\left[L_{\mu}(\hat W)-\ell(\hat W,\tilde{Z})\right]\label{eqnAA}
\end{align}
where $\tilde{Z}\in\mathcal{Z}$ is distributed according to $\mu'$.  The above bound has a similar form as the one given in Theorem \ref{Thm2}. Observe that \eqref{eqnAA} provides a generalization bound on the algorithm $\A$ based on the sole assumption that it uses a training data of size $n$. If more is known about the algorithm, \emph{e.g.} an upper bound on the input and output mutual information, we can write better bounds as follows:

\begin{theorem}\label{thm11}
Let
\begin{equation}
\tilde{\D}_2(r)\triangleq \max_{P_{\hat W|\tilde{Z}}:~I(\hat W;\tilde{Z})\leq r,~~\mathbb{E}[\tilde\ell(\hat W,\tilde{Z})]\geq v_n}
\mathbb{E}\left[L_{\mu}(\hat W)-\ell(\hat W,\tilde{Z})\right].
\end{equation}
Then,
\[
\D_{1}(r)\le\tilde{\D}_2(r/n)\le {\D}_2(r/n).
\]
\end{theorem}
Proof of the Theorem \ref{thm11} can be found in Section \ref{appenEE}.

\begin{example}
Consider the setting in Example \ref{ex1}. Figure \ref{fig:fig2} illustrates this improvement in $\D_2(r/n)$ when  $\tilde{\ell}(w,z)=-\mathbf{1}[w\neq z]$ and $n=10$.
\end{example}
\begin{example}
Consider the setting in Example \ref{ex2}. Figure \ref{fig:fig2} illustrates this improvement in $\D_2(r/n)$ when  $\tilde{\ell}(w,z)=(w-z)^{2}$ and $n=10$.
\end{example}

\begin{figure}
	   \centering
	\includegraphics[scale=1,width=0.8\linewidth]{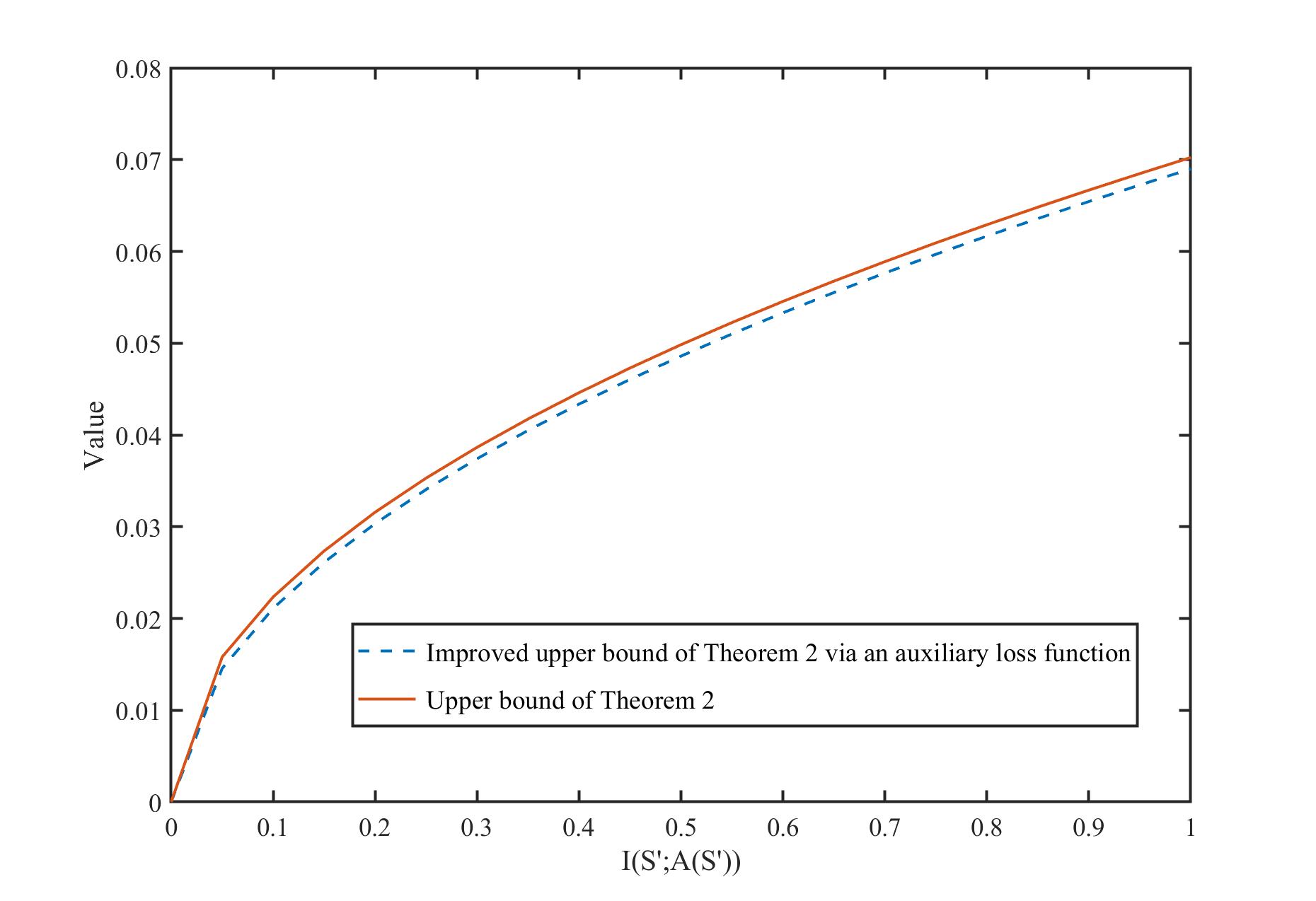}
	\caption{The bound in  Theorem \ref{Thm2} and its improved version via the auxiliary loss function $\tilde{\ell}(w,z)=-\mathbf{1}[w\neq z]$ for $\mathcal{W}=\mathcal{Z}=\{0,1\}$ and $n=10$ and the original loss function $\ell(w,z)=w\cdot z$.}
	\label{fig:fig2}
\end{figure}
\begin{figure}
    \centering
    \includegraphics[scale=1,width=0.8\linewidth]{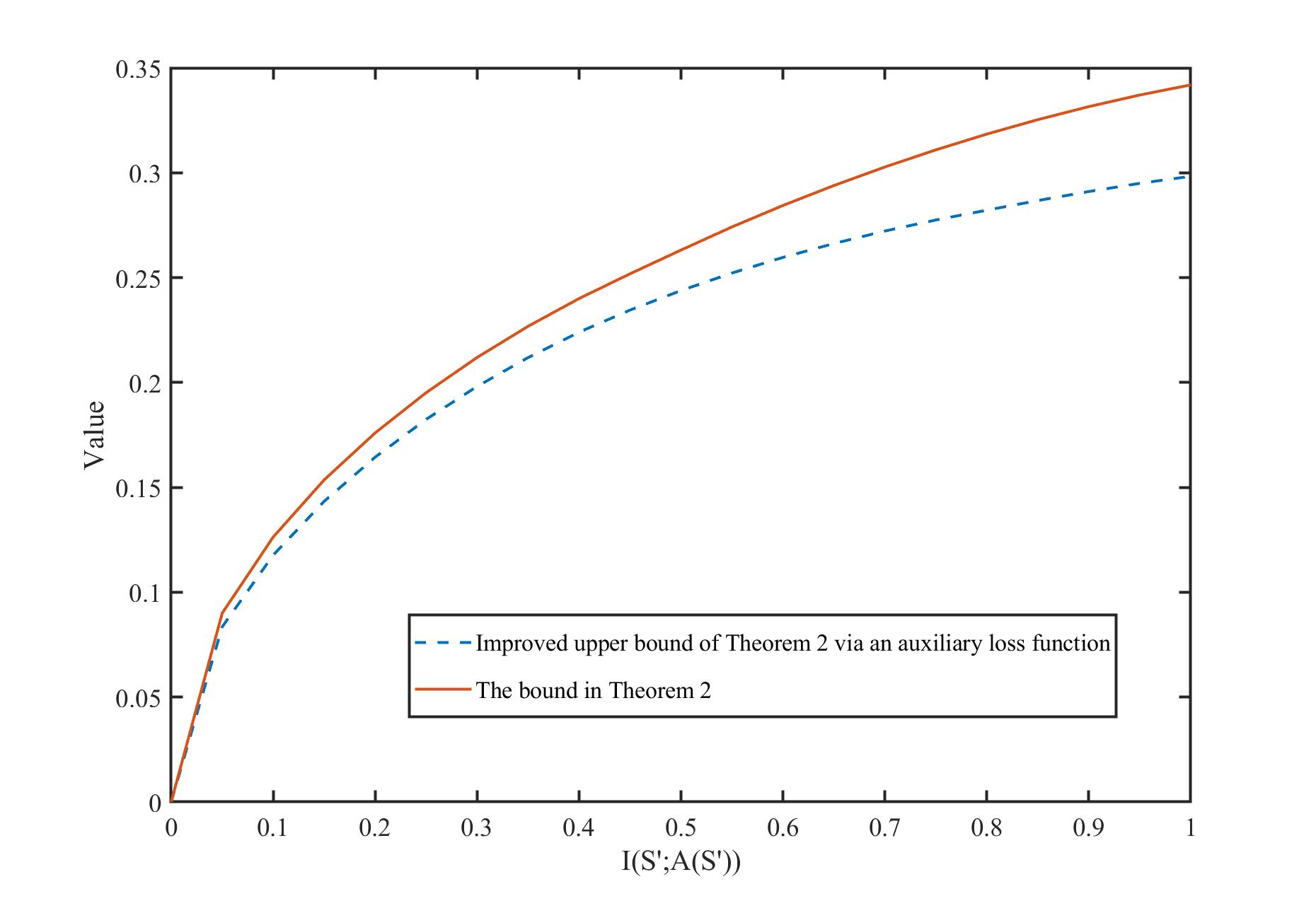}
    \caption{The bound in  Theorem \ref{Thm2} and its improved version via the auxiliary loss function $\tilde{\ell}(w,z)=(w-z)^{2}$ for the learning setting $\mathcal{W}=[0,1],\,\mathcal{Z}=\{0,1\}$ and $n=10$ and the original loss function $\ell(w,z)=|w-z|$.}
    \label{fig:fig2new}
\end{figure}

In order to use the bound in Theorem \ref{thm11}, one must know the value of $v_n$. However, this is not known in practice. For instance, consider the special case of loss function $\tilde{\ell}(w,z)=(w-z)^{2}$. Given a training data $(z'_1, z'_2, \cdots, z'_n)$, the output of the ERM algorithm with the quadratic loss is just the average of the traning data samples and $v_n$ equals \[\dfrac{n-1}{n}\mathsf{Var}_{\mu'}(Z').\]
The variance of the test data is not known, but can be estimated from the training dataset itself. Below we show how to estimate  $v_n$ by running the ERM algorithm on the available training data.  
Assume that the auxiliary loss satisfies $|\tilde{\ell}(w,z)-\tilde{\ell}(w,z')|\leq c$ for all $w,z,z'$. Then, we have
\begin{align*}
    \mathsf{ERM}(z'_1,z'_2,\cdots,z'_n)&=\min_{w}\frac1n\sum_{i=1}^n\tilde\ell(w,z_i)
    \\&\leq \min_{w}\left[\frac{c}{n}+
    \frac1n\tilde\ell(w,z''_1)+
    \frac1n\sum_{i=2}^n\tilde\ell(w,z'_i)
    \right]
    \\&=\frac{c}{n}+\mathsf{ERM}(z''_1,z'_2,\cdots,z'_n).
\end{align*}
Then McDiarmid's inequality implies high concentration around expected value for the ERM algorithm:
$$\mathbb{P}\left[\big|\mathsf{ERM}-\mathbb{E}[\mathsf{ERM}]\big|\geq t\right]\leq 2e^{-\frac{2nt^2}{c^2}}.
$$
Thus, one can find an estimate for $v_n$ with high probability based on the available training data sequence.

At the end, we remark that it is also possible to write bounds based on multiple auxiliary loss functions rather than just one.

%%%%%%%%%%%%%%%%%%%%%%
\section{Acknowledgment}
The first author is also grateful to Dr.\,\,Mohammad Mahdi Mojahedian for helpful discussions on learning from heterogeneous data in mixture models which gave birth to some ideas in this work.
%%%%%%%%%%%%%%%%%%%%%%
%%%% Bibliography %%%%

%%%%%%%%%%%%%%%%%%%%%%%%%%%%%%%%%%%%%%%%%%%%%%%%%
%%%%%%%%%%%%%%%%%%%%%%%%%%%%%%%%%%%%%%%%%%%%%%%%%%%%%%%%%%%

\section{Proofs of the results}
\label{sec:proofs}
In the following sections we present the proofs of the results stated in the previous section in their order of appearance.

\subsection{Proof of Theorem \ref{Thm2}}
\label{ProofThm2}
Let $\tw=(\tilde{w}_1, \tilde{w}_2, \cdots, \tilde{w}_n)\in\mathcal{W}^n$ be a sequence of length $n$. Let
\begin{align}
\bar{\D}_1(r)&\triangleq \sup_{P_{\tW|S'}:~I(\tW;S')\leq r}
\frac1n\sum_{i=1}^n\mathbb{E}\left[L_{\mu}(\tilde{W}_i)-\ell(\tilde{W}_i,Z'_i)\right]\label{tildeD}
\end{align}
where $S'=(Z'_1, Z'_2, \cdots, Z'_n)$ and
$\tW=(\tilde{W}_1, \tilde{W}_2, \cdots, \tilde{W}_n)$. Observe that if the entries of the vector $\tW$ are all equal, the expression in \eqref{tildeD} reduces to the one  in \eqref{eqnRD1}. Therefore, in \eqref{tildeD} we are taking the supremum over a larger set. Thus, $\bar{\D}_1(r)\geq {\D}_1(r)$. It follows that for any algorithm $\A$ satisfying $I\big(S';\A(S')\big)\leq r$, we have
\[
\mathrm{gen}\left(\mu,\mu',\A\right)
\leq \bar{\D}_1(r).\]
We claim that
$\bar{\D}_1(r)=\D_2(r/n)$. The proof follows similar steps as in \cite[Section 3.6.2]{el2011network} for lossy compression. However, we provide a proof for completeness. We first claim that $\bar{\D}_1(r)\geq \D_2(r/n)$. To see this, take some $P_{\hat W|Z'}$  in \eqref{eqnRD22} and take 
\[p(\tw|s')=\prod_{i=1}^np_{\hat W|Z'}(\tilde{w}_i|z'_i).
\]
This special choice for $P_{\hat W|Z'}$ in \eqref{tildeD} shows that $\bar{\D}_1(r)\geq \D_2(r/n)$. 

It remains to show that $\bar{\D}_1(r)\leq \D_2(r/n)$. Take some arbitrary $P_{\tW|S'}$ satisfying $I(\tW;S')\leq r$. We have
\begin{align}
    r&\geq I(\tW;S')\nonumber
    \\&=\sum_{i}I(\tW;Z'_i|Z'^{i-1})\nonumber
    \\&=\sum_{i}I(\tW,Z'^{i-1};Z'_i)\label{eqnRR1}
    \\&\geq \sum_{i}I(\tilde{W}_i;Z'_i)\label{eqnRR12}
\end{align}
where \eqref{eqnRR1} follows from the fact that $Z'_i$ are i.i.d. random variables. We also have
\begin{align}
    \frac1n\sum_{i=1}^n\mathbb{E}\left[L_{\mu}(\tilde{W}_i)-\ell(\tilde{W}_i,Z'_i)\right]&\leq 
    \frac1n\sum_{i=1}^n\D_2(I(\tilde{W}_i;Z'_i))\label{eqnpp1}
    \\&
    \leq 
    \D_2\left(\frac1n\sum_{i=1}^nI(\tilde{W}_i;Z'_i)\right)\label{eqnpp2}
    \\&\leq \D_2\left(\frac{r}{n}\right)\label{eqnpp3}
\end{align}
where \eqref{eqnpp1} follows from
the definition of $\D_2$, \eqref{eqnpp2} follows from concavity of $\D_2(\cdot)$, and \eqref{eqnpp3}
follows from \eqref{eqnRR12} and the fact that $\D_2(\cdot)$ is an increasing function. Concavity of $\D_2(\cdot)$ follows from the fact that mutual information $I(\tW;S')$ is convex in $P_{\tW|S'}$ for a fixed distribution on $P_{S'}$.

Since $P_{\tW|S'}$ was an arbitrary conditional distribution satisfying $I(\tW;S')\leq r$, we deduce from \eqref{eqnpp1}-\eqref{eqnpp3} that $\D_2(r/n)\geq\bar{\D}_1(r)$ as desired.

The cardinality bounds on the auxiliary random variable $\hat{W}$ in the definition of $\D_2$ comes from the standard Caratheodory-Bunt \cite{bun34} arguments and is omitted.

\subsection{Proof of Theorem \ref{th3}}
\label{appenB}

Given the distribution $\zeta(x)$ and some arbitrary conditional distribution $\zeta(\hat{x}|x)$, let $\zeta(x,\hat{x})=\zeta(\hat{x}|x)\zeta(x)$. Set $q(x,\hat{x})= \eta(x)\zeta(\hat{x})$ and $f(x,\hat{x})=\lambda d(x,\hat{x})$ where $-b<\lambda<0$.
From the Donsker-Varadhan representation, we obtain that 
\begin{align}\label{eq1234}
D(\zeta_{X,\hat{X}}\|q_{X,\hat{X}})\ge\lambda\mathbb{E}_{\zeta}[d(\hat{X},X)]-\log\mathbb{E}_{q}\left[e^{\lambda(d(X,\hat{X}))}\right],
\end{align}
Using independence of $X$ and $\hat{X}$ under $q$ we can write for $-b<\lambda<0$,
\[
\log\mathbb{E}_{q}\left[e^{\lambda(d(X,\hat{X}))}\right]
=
\log\mathbb{E}_{\hat X\sim\zeta}\left\{\mathbb{E}_{X\sim\eta}\left[e^{\lambda(d(X,\hat{X}))}\right]\right\}
\leq \sup_{\hat{x}}\log \mathbb{E}_{\eta}\left[e^{\lambda d(X,\hat{x})}\right]\leq \phi(\lambda).
\]
Then from \eqref{eq1234} and consider $\lambda<0$,
\[
\mathbb{E}_{\zeta}[d(\hat{X},X)]\ge \frac{1}{\lambda}D(\zeta_{X,\hat{X}}\|q_{X,\hat{X}})+\frac{1}{\lambda}\phi(\lambda).
\]

Moreover, $D(\zeta_{X,\hat{X}}\|q_{X,\hat{X}})=I_{\zeta}(\hat{X};X)+D(\zeta_X\|\eta_X)$ and $I_{\zeta}(\hat{X};X)\le r$. Thus,
\[
\mathbb{E}_{\zeta}[d(\hat{X},X)]\ge \sup_{-b<\lambda<0}\left\lbrace\frac{1}{\lambda}\left[r+D(\zeta_X\|\eta_X)\right]+\frac{1}{\lambda}\phi(\lambda)\right\rbrace.
\]
In conclusion,
 \begin{align}
      \inf_{P_{\hat{X}|X}:\,I_{\zeta}(\hat{X};X)\le r}\mathbb{E}_{ \zeta}\left[d(X,\hat{X})\right]\ge \sup_{-b<\lambda<0}\left\lbrace\frac{1}{\lambda}\left[r+D(\zeta_X\|\eta_X)\right]+\frac{1}{\lambda}\phi(\lambda)\right\rbrace.
    \end{align}
\subsection{Proof of Theorem \ref{Thm2b2}}
\label{appenThm2b2}
The inequality
\[\frac1n\sum_{i=1}^n \D_2\left(I(Z'_i;\A(S'))\right)\le \frac{1}{n}\sum_{i=1}^{n}\sqrt{{2\sigma^{2}}\left[I\big(Z_{i}';\A(S')\big)+D(\mu'\|\mu)\right]}.
\]
follows from Corollary \ref{cor1}. To show the inequality
\[
\mathrm{gen}\left(\mu,\mu',\A\right)\le \frac1n\sum_{i=1}^n \D_2\left(I(Z'_i;\A(S'))\right)
\]
take some algorithm $\A$ and let $W'=\A(S'^n)$. Then, 
\begin{align}
   \mathrm{gen}\left(\mu,\mu',\A\right)&= \frac1n\sum_{i=1}^n\mathbb{E}\left[L_{\mu}(W')-\ell(W',Z'_i)\right]\leq 
    \frac1n\sum_{i=1}^n\D_2(I(W';Z'_i))\label{eqnpp1dd}
\end{align}
where \eqref{eqnpp1dd} follows from
the definition of $\D_2$.

\subsection{Proof of Theorem \ref{thm4}}
\label{Sec:ProofThm4}

It suffices to prove the lower bound when a fixed distribution $P_{W'}\in \mathcal{M}$ is chosen for the output of the algorithm because a minimum can be taken over all $P_{W'}\in \mathcal{M}$ from both sides of the desired inequality at the end. We have
\begin{align}
	\D_3(r)=&\min_{{P_{W'S'}\in U(P_{W'},P_{S'})}: I(W';S')\le r }\mathbb{E}\left[L_{\mu}(W')-L_{S'}(W')\right]
	\\
	&=
	\min_{{P_{W'S'}\in U(P_{W'},P_{S'})}}\max_{\lambda\geq 0}\mathbb{E}\left[L_{\mu}(W')-L_{S'}(W')\right]+\lambda D(P_{W'S'}\|P_{W'}P_{S'})-\lambda r
	\\
	&\geq
	\max_{\lambda\geq 0}\min_{{P_{W'S'}\in U(P_{W'},P_{S'})}}\mathbb{E}\left[L_{\mu}(W')-L_{S'}(W')\right]+\lambda D(P_{W'S'}\|P_{W'}P_{S'})-\lambda r
	\\
	&\geq 
	\max_{\lambda\geq 0}\left[\D_3(0)+\lambda-\lambda\psi(1/(n\lambda))^{n}-\lambda r\right]\label{eqn:frl}
	\\
	&=\D_3(0)-
	\min_{\lambda\geq 0}\left[\lambda r+\lambda (\psi(1/(n\lambda))^{n}-1)\right].
\end{align}
where \eqref{eqn:frl} follows from Lemma \ref{lemma:frl}.

\begin{lemma}\label{lemma:frl}
\label{lowerbound_on_gen}
	Let $\ell(W',z)$ satisfies
	 \begin{align*}
     \psi(\lambda)\geq \mathbb{E}_{P_{W'}}\left[e^{\lambda\left[\ell(W',z)-\mathbb{E}_{P_{W'}}[\ell(W',z)]\right]}\right], \quad\forall z\in \mathcal{Z}.
\end{align*}
	Then, for any $\lambda\geq 0$
	\begin{align}
	\min_{{P_{W'S'}\in U(P_{W'},P_{S'})}}\mathbb{E}_{P_{W'S'}}\left[L_{\mu}(W')-L_{S'}(W')\right]+\lambda D(P_{W'S'}\|P_{W'}P_{S'})
	\geq 
	\D_3(0)-
    \lambda (\psi(1/(n\lambda))^{n}-1).\label{eqnLemmal1}
	\end{align}
\end{lemma}
\begin{proof}Assume that $W'\sim \zeta$ and $S'\sim \beta$ are the marginal distributions of $W'$ and $S'$. Setting  $$\Delta(w,s)=L_{\mu}(w)-L_{s}(w),$$ we can express the left hand side of \eqref{eqnLemmal1} as
\begin{align}
	\min_{(W',S')\sim\pi\in U(\zeta,\beta)}\mathbb{E}_{\pi}\Delta(W',S')+\lambda \int_{\mathcal{W}\times \mathcal{Z}^{\otimes n}}\phi\left(\frac{d\pi(w',s')}{d\zeta(w')d\beta(s')}\right)d\zeta(w')d\beta(s')
\end{align} 
where $\phi(x)=x\log(x)-x+1$.
We find the dual problem of the above optimization problem. 
Introducing the Lagrange multipliers $f$
and $g$ associated to the constraints, the Lagrangian reads
\begin{align*}
	\mathcal{L}(\lambda,\zeta,\beta)=&\mathbb{E}_{\pi}\Delta(W',S')+\lambda \int_{\mathcal{W}\times \mathcal{Z}^{\otimes n}}\phi\left(\frac{d\pi(w',s')}{d\zeta(w')d\beta(s')}\right)d\zeta(w')d\beta(s')
	\\&+\int_{\mathcal{W}}f(w')\left(d\zeta(w')-\int_{\mathcal{Z}^{\otimes n}}d\pi(w',s')\right)\\&+\int_{\mathcal{Z}^{\otimes n}}g(s')\left(d\beta(s')-\int_{\mathcal{W}}d\pi(w',s')\right).
\end{align*}
The dual Lagrange function is given by $\min_{\pi}	\mathcal{L}(\lambda,\alpha,\beta)$ over all $\pi(w',s')\geq 0$. Note that in computing the minimum we do not require $\sum_{w',s'}\pi(w',s')=1$. Observe that
\begin{align*}
	&\min_{\pi}	\mathcal{L}(\lambda,\zeta,\beta)
	\\&=\int_{\mathcal{W}}f(w')d\zeta(w')+\int_{\mathcal{Z}^{\otimes n}}g(s')d\beta(s')\\
	&+\lambda\min_{\pi}\left(\int_{\mathcal{W}\times \mathcal{Z}^{\otimes n}}\left(\phi\left(\frac{d\pi(w',s')}{d\zeta(w')d\beta(s')}\right)+\frac{\Delta(w',s')-f(w')-g(s')}{\lambda}\frac{d\pi(w',s')}{d\zeta(w')d\beta(s')}\right)d\zeta(w')d\beta(s')\right)\\
	&=\int_{\mathcal{W}}f(w')d\zeta(w')+\int_{\mathcal{Z}^{\otimes n}}g(s')d\beta(s')-\lambda\int_{\mathcal{W}\times \mathcal{Z}^{\otimes n}}\phi^{*}\left(\frac{f(w')+g(s')-\Delta(w',s')}{\lambda}\right)d\zeta(w')d\beta(s'),
\end{align*}
where $\phi^{*}$ is the Legendre transform of $\phi$ given by
\[
\phi^{*}(y)=\sup_{x\geq 0}\big[xy-\phi(x)\big]=e^{y}-1.
\]
Thus, we obtain
\begin{align}
\min_{\pi}	\mathcal{L}(\lambda,\zeta,\beta)=
	\mathbb{E}[g(S')]+\mathbb{E}[f(W')]+\lambda-\lambda\mathbb{E}_{P_{S'}P_{W'}}\left[\exp\left(-\frac{\Delta(W',S')-g(S')-f(W')}{\lambda}\right)\right].
\end{align}
From weak duality, for every continuous functions $f$ and $g$,
\begin{align}
	&\mathbb{E}[g(S')]+\mathbb{E}[f(W')]+\lambda-\lambda\mathbb{E}_{P_{S'}P_{W'}}\left[\exp\left(-\frac{\Delta(W',S')-g(S')-f(W')}{\lambda}\right)\right]
	\nonumber\\
	&\le \min_{{P_{W'S'}\in U(P_{W'},P_{S'})}}\mathbb{E}_{P_{W'S'}}\left[L_{\mu}(W')-L_{S'}(W')\right]+\lambda D(P_{W'S'}\|P_{W'}P_{S'}).	
\end{align}
Assigning $g(s')=-\mathbb{E}_{W'\sim\zeta}L_{s'}(W')$  and $f(w')=L_{\mu}(w')$ and using the fact that 
$\Delta(w,s)=L_{\mu}(w)-L_{s}(w)$
, we obtain
\begin{align}
	&\mathbb{E}[L_{\mu}(W')-L_{\mu'}(W')]+\lambda-\lambda\mathbb{E}_{P_{S'}P_{W'}}\left[\exp\left(\frac{L_{S'}(W')-\mathbb{E}_{P_{W'}}L_{S'}(W')}{\lambda}\right)\right]
	\nonumber\\
	&\le \min_{{P_{W'S'}\in U(P_{W'},P_{S'})}}\mathbb{E}_{P_{W'S'}}\left[L_{\mu}(W')-L_{S'}(W')\right]+\lambda D(P_{W'S'}\|P_{W'}P_{S'}).
\end{align}
We give an upper bound for the exponential term as follows:
\begin{align*}
    &\mathbb{E}_{P_{S'}P_{W'}}\left[\exp\left(\frac{L_{S'}(W')-\mathbb{E}_{P_{W'}}L_{S'}(W')}{\lambda}\right)\right]= \prod_{i=1}^{n}\mathbb{E}_{P_{Z'_{i}}P_{W'}}\left[\exp\left(\frac{\ell(W',Z'_{i})-\mathbb{E}_{P_{W'}}\ell(W',Z'_{i})}{n\lambda}\right)\right]\nonumber\\
    &\le \prod_{i=1}^{n}\sup_{z'\in\mathcal{Z}}\mathbb{E}_{P_{W'}}\left[\exp\left(\frac{\ell(W',z')-\mathbb{E}_{P_{W'}}\ell(W',z')}{n\lambda}\right)\right]\le \left(\psi\left(\frac{1}{n\lambda}\right)\right)^{n}.
\end{align*}
Thus,
\begin{align*}
	&\min_{{P_{W'S'}\in U(P_{W'},P_{S'})}}\mathbb{E}_{P_{W'S'}}\left[L_{\mu}(W')-L_{S'}(W')\right]+\lambda D(P_{W'S'}\|P_{W'}P_{S'})
	 \nonumber\\
&\geq\mathbb{E}[L_{\mu}(W')-L_{\mu'}(W')]+\lambda-\lambda \left(\psi\left(\frac{1}{n\lambda}\right)\right)^{n}=
\D_3(0)+\lambda-\lambda \left(\psi\left(\frac{1}{n\lambda}\right)\right)^{n}.
\end{align*}
\end{proof}
%%%%%%%%%%%
%%%%%%%%%%%

\subsection{Proof of Theorem \ref{Thm2b}}
\label{appenDE}
The proof is similar to the proof of Theorem \ref{Thm2}. As in Section
\ref{ProofThm2}, we let $\tw=(\tilde{w}_1, \tilde{w}_2, \cdots, \tilde{w}_n)\in\mathcal{W}^n$ be a sequence of length $n$. Let
\begin{align}
\tilde{\D}_3(r)&\triangleq \inf
\frac1n\sum_{i=1}^n\mathbb{E}\left[L_{\mu}(\tilde{W}_i)-\ell(\tilde{W}_i,Z'_i)\right]\label{eqnRD222}
\end{align}
where the infimum is over $P_{\tW,S'}$ satisfying $ P_{\tilde{W}_i,S'}\in U(P_{W'},P_{S'})$ and $I(\tW;S')\leq r$. Observe that if the entries of the vector $\tW$ are all equal, the expression in \eqref{eqnRD222} reduces to the one  in \eqref{eqnRD22b}. Therefore, in \eqref{eqnRD222} we are taking the infimum over a larger set. Thus, $\tilde{\D}_3(r)\leq {\D}_3(r)$. It follows that for any algorithm $\A$ satisfying $I\big(S';\A(S')\big)\leq r$, we have
\[
\mathrm{gen}\left(\mu,\mu',\A\right)
\geq \tilde{\D}_3(r).\]
We claim that
$\tilde{\D}_3(r)=\D_3(r/n)$. The rest of the proof follows similar lines as in the proof of Theorem \ref{Thm2} given in Section
\ref{ProofThm2}. Thus, it is omitted.

%%%%%%%%%%%%%%%%%%%%%%%%
%%%%%%%%%%%%%%%%%%%

\subsection{Proof of Theorem \ref{lhighprobability_on_gen}}
\label{Sec:highProb}

We would like to bound ${L}_{\mu}(\A(S'))-L_{S'}(\A(S'))$ from above. Let $S=(Z_1, Z_2, \cdots, Z_n)$ be distributed according to $\mu^{\otimes n}$, while $S'=(Z'_1, Z'_2, \cdots, Z'_n)$ was distributed according to $(\mu')^{\otimes n}$. From the subgaussian assumption, we have
\begin{align}
	\mathbb{E}_{P_{S}P_{W'}}\left[\exp\left(\lambda L_{\mu}(W')-\lambda {L}_{S}(W')\right)\right]\le \exp\left(\frac{\lambda^{2}\sigma^{2}}{2n}\right).
\end{align}
Using a change of measure argument, we obtain \begin{align}
	\mathbb{E}_{P_{W'S'}}\left[\exp\left(\lambda L_{\mu}(W')-\lambda {L}_{S'}(W')-\frac{\lambda^{2}\sigma^{2}}{2n}-\log\frac{dP_{W'S'}}{dP_{W'}dP_{S}}\right)\right]\le 1.
\end{align}
Using Markov's inequality $\mathbb{P}[X>\frac{1}{\delta}]<\mathbb{E}[X]\delta$, we deduce
\begin{align}
	\mathbb{P}_{W'S'}\left[\exp\left(\lambda L_{\mu}(W')-\lambda {L}_{S'}(W')-\frac{\lambda^{2}\sigma^{2}}{2n}-\log\frac{dP_{W'S'}}{dP_{W'}dP_{S}}\right)\ge \frac{1}{\delta}\right]\le \delta.
\end{align}
Equivalently,
\begin{align}
	\mathbb{P}_{W'S'}\left[\lambda L_{\mu}(W')-\lambda {L}_{S'}(W')\ge \frac{\lambda^{2}\sigma^{2}}{2n}+\log\frac{dP_{W'S'}}{dP_{W'}dP_{S}}+\log\left(\frac{1}{\delta}\right)\right]\le \delta.
\end{align}
Thus, the following inequality holds with probability at least $1-\frac{\delta}{2}$:
\begin{align}\label{eq99}
	L_{\mu}(W')- {L}_{S'}(W')\le_{1-\frac{\delta}{2}} \frac{\lambda\sigma^{2}}{2n}+\frac{1}{\lambda}\log\frac{dP_{W'S'}}{dP_{W'}dP_{S}}+\frac{1}{\lambda}\log\left(\frac{2}{\delta}\right).
\end{align}
Using Chernoff's bound on $\log\left(\frac{dP_{W'S'}}{dP_{W'}dP_{S'}}\right)$, we get for $1<\alpha$,
\[
\mathbb{P}_{W'S'}\left[\log\left(\frac{dP_{W'S'}}{dP_{W'}dP_{S}}\right)\ge t\right]\le\frac{\mathbb{E}_{P_{W'S'}}\left[e^{(\alpha-1)\log\left(\frac{dP_{W'S'}}{P_{W'}P_{S}}\right)}\right]}{e^{(\alpha-1)t}}=e^{(\alpha-1)\left(D_{\alpha}(P_{W'S'}\|P_{W'}P_{S})-t\right)}.
\]

Thus,
\[
\log\left(\frac{dP_{W'S'}}{dP_{W'}dP_{S}}\right)\le_{1-\frac{\delta}{2}}D_{\alpha}(P_{W'S'}\|P_{W'}P_{S}) +\frac{1}{\alpha-1}\log\left(\frac{2}{\delta}\right).
\]
For the case of no-mismatch, the above equation together with \eqref{eq99} recovers the result of \cite{esposito2019generalization} once we optimize  over $\lambda$. 

We use Lemma \ref{lemm4} to show the following inequality:
\[
D_{\alpha}(P_{W'S'}||P_{W'}P_S)\leq D_{1+(\alpha-1)p}(P_{W'S'}\|P_{W'}P_{S'})+nD_{1+(\alpha-1)q}(\mu'\|\mu),
\]
where $p$ and $q$ are non-negative and Holder conjugate. Therefore, combining with \eqref{eq99}, we get
\begin{align*}
	L_{\mu}(W')- {L}_{S'}(W')\le_{1-\delta} \frac{\lambda\sigma^{2}}{2n}+\frac{1}{\lambda}D_{1+(\alpha-1)p}(P_{W'S'}\|P_{W'}P_{S'})+\frac{n}{\lambda}D_{1+(\alpha-1)q}(\mu'\|\mu)+\frac{\alpha}{(\alpha-1)\lambda}\log\left(\frac{2}{\delta}\right).
\end{align*} 
Optimizing over $\lambda$ yields
\begin{align}\label{1000}
	L_{\mu}(W')- {L}_{S'}(W')\le_{1-\delta} \sqrt{2\sigma^{2}D_{1+(\alpha-1)q}(\mu'\|\mu)+\frac{2\sigma^{2}\left[D_{1+(\alpha-1)p}(P_{W'S'}\|P_{W'}P_{S'})+\frac{\alpha}{\alpha-1}\log\left(\frac{2}{\delta}\right)\right]}{n}}.
\end{align}

Then with $\alpha=\frac{3}{2}$ and $p=q=2$, we get from equation \eqref{1000},
\[
 \mathbb{P}[|\mathrm{gen}_{\mu}(W',S')|\geq \eta]\leq 2\exp\left(-\frac{n\left(\frac{\eta^{2}}{2}-\sigma^{2}D_{2}(\mu'\|\mu)\right)-\sigma^{2}D_{2}(P_{W'S'}\|P_{W'}P_{S'})}{3\sigma^{2}}\right).
\]
\begin{lemma}
For $1<\alpha,p,q<\infty$ with $\frac{1}{p}+\frac{1}{q}=1$, we have
\[
D_{\alpha}(P_{W'S'}||P_{W'}P_S)\leq D_{1+(\alpha-1)p}(P_{W'S'}\|P_{W'}P_{S'})+nD_{1+(\alpha-1)q}(\mu'\|\mu).
\]\label{lemm4}
\end{lemma}
\begin{proof}
We use Holder's inequality for $1<p,q<\infty$ in following inequality:
\begin{align*}
    &\exp\left((\alpha-1)D_{\alpha}(P_{W'S'}||P_{W'}P_S)\right)=\int\left(\frac{dP_{W'S'}}{dP_{W'}dP_S}\right)^{\alpha-1}dP_{W'S'}
    \\&=\int\left(\frac{dP_{W'S'}}{dP_{W'}dP_S'}\frac{dP_S'}{dP_S}\right)^{\alpha-1}dP_{W'S'}\nonumber\\
    &\le \left(\int\left(\frac{dP_{W'S'}}{dP_{W'}dP_S'}\right)^{p(\alpha-1)}dP_{W'S'}\right)^{\frac{1}{p}}\left(\int\left(\frac{dP_S'}{dP_S}\right)^{q(\alpha-1)}dP_{W'S'}\right)^{\frac{1}{q}}
    \\
    &=\exp\left((\alpha-1)D_{1+p(\alpha-1)}(P_{W'S'}\|P_{W'}P_{S'})\right)\exp\left((\alpha-1)D_{1+q(\alpha-1)}(P_{S'}\|P_{S})\right).
\end{align*}
Then we get,
\[
D_{\alpha}(P_{W'S'}||P_{W'}P_S)\le D_{1+p(\alpha-1)}(P_{W'S'}\|P_{W'}P_{S'})+nD_{1+q(\alpha-1)}(\mu'\|\mu).
\]
\end{proof}
%%%%%%%%%%%%%%%%%%%%%%%%

\subsection{Proof of Theorem \ref{thm8}}
\label{sec:proof8}
\begin{proof}
Let $w_{0}=\argmin_{w\in\mathcal{W}}L_{\mu}(w)$. Let $\A^*$ be an algorithm that outputs $w_0$ regardless of the training data sequence. Then, using Theorem \ref{lhighprobability_on_gen} with $\A^*$ we obtain
\begin{align}
    L_{S'}(w_{0})-L_{\mu}(w_{0})\le_{1-\delta}\sqrt{2\sigma^{2}D_2(\mu'\|\mu)+2\sigma^{2}\frac{\log\left(\frac{2}{\delta}\right)}{n}}.
\end{align}
On the other hand, from the definition of the ERM algorithm we have
\begin{align}
    L_{S'}(\A_{\mathrm{ERM}}(S'))\le L_{S'}(w_{0}).
\end{align}
Since $L_{\mu}(w_{0})=\min_{w\in\mathcal{W}}L_{\mu}(w)$, it follows that
\begin{align}
    L_{S'}(\A_{\mathrm{ERM}}(S'))\le_{1-\delta} \min_{w\in\mathcal{W}}L_{\mu}(w)+\sqrt{2\sigma^{2}D_2(\mu'\|\mu)+2\sigma^{2}\frac{\log\left(\frac{2}{\delta}\right)}{n}}.
\end{align}
Then using Theorem \ref{lhighprobability_on_gen},
\begin{align}
    &L_{\mu}(\A_{\mathrm{ERM}}(S'))-\min_{w\in\mathcal{W}}L_{\mu}(w)=L_{\mu}(\A_{\mathrm{ERM}}(S'))-L_{S'}(\A_{\mathrm{ERM}}(S'))+L_{S'}(\A_{\mathrm{ERM}}(S'))-\min_{w\in\mathcal{W}}L_{\mu}(w)\nonumber\\
    &\le_{1-\delta}\sqrt{2\sigma^{2}D_2(\mu'\|\mu)+\frac{2\sigma^{2}\left[D_{2}(P_{W'S'}\|P_{W'}P_{S'})+2\log\left(\frac{4}{\delta}\right)\right]}{n}}+\sqrt{2\sigma^{2}D_2(\mu'\|\mu)+2\sigma^{2}\frac{\log\left(\frac{4}{\delta}\right)}{n}}.
\end{align}
\end{proof}

%%%%%%%%%%%%%%%%%%%%%%%%%%%%%%%%%%%%%%%%%%%%%
%%%%%%%%%%%%%%%%%%%%%%%%%%%%%%%%%%%%%%%%%%%%%%%%%%%%%%%%%
%%%%%%%%%%%%%%%%%%%%%%%%%%%%%%%%%%%%%%%%%%%%%%%%%%%%%%%%%
\subsection{Proof of Theorem \ref{thm:sample_mix_comp}}
\label{appenDD}

Take some arbitrary $\mu\in \mathcal{P}$ and $\mu'\in\mathcal{P}_\gamma$. It suffices to find a bound on the difference
$\epsilon(\A, \mu', n, {\delta})-\epsilon(\A,  \mu, n, \delta)$ that depends only on $D(\mu'\|\mu)$. 
From Lemma \ref{lmm2}, given the training data $S'=(Z'_1, \cdots, Z'_n)\sim (\mu')^{\otimes n}$, we can define $S=(Z_1, \cdots, Z_n)\sim \mu^{\otimes n}$ such that $(Z_i,Z'_i)$ are i.i.d. for $1\le i\le n$ and
\begin{align}\label{coupling}
\mathbb{P}\left[Z_i\neq Z'_i\right]=\|\mu-\mu'\|_{TV}.
\end{align}

 For the first upper bound \eqref{eq100}, we write
\begin{align}\label{eq600}
    &L_{\mu'}(\A(S'))=L_{\mu'}(\A(S')-L_{\mu'}(\A(S))+L_{\mu}(\A(S))+L_{\mu'}(\A(S))-L_{\mu}(\A(S))\nonumber\\
    &\overset{(a)}{\le}_{1-\delta} \sum_{i=1}^{n}\beta_{i}+\min_{w\in\mathcal{W}}L_{\mu}(w)+\epsilon(\A,  \mathcal{P}, n, \delta)+L_{\mu'}(\A(S))-L_{\mu}(\A(S))\nonumber\\
    &\overset{(b)}{\le}\sum_{i=1}^{n}\beta_{i}+\min_{w\in\mathcal{W}}L_{\mu}(w)+\epsilon(\A,  \mathcal{P}, n, \delta)+\sqrt{2\sigma^{2}\gamma},
\end{align}
where, (a) comes from the uniform stability condition and definition of $\epsilon(\A,  \mathcal{P}, n, \delta)$. Inequality (b) is derived using Lemma \ref{lemma1} as follows
\begin{align*}
    L_{\mu'}(\A(S))-L_{\mu}(\A(S))&=\mathbb{E}_{\mu'}\left[\ell(\A(S),Z')-\mathbb{E}_{\mu}[\ell(\A(S),Z)]\right]\\
    &\le \frac{1}{\lambda}D(\mu'\|\mu)+\frac{1}{\lambda}\log\mathbb{E}_{\mu}\left[e^{\lambda\left[\ell(\A(S),Z')-\mathbb{E}_{\mu}[\ell(\A(S),Z)]\right]}\right]\\&\le \frac{1}{\lambda}D(\mu'\|\mu)+\frac{\lambda\sigma^{2}}{2}.
\end{align*}
Optimizing on $\lambda$ and $D(\mu'\|\mu)\le \gamma$, we get
\begin{align}\label{eq9911}
     L_{\mu'}(\A(S))-L_{\mu}(\A(S))\le \sqrt{2\sigma^{2}\gamma}.
\end{align}
Next, we give an upper bound for $\min_{w}{L}_{\mu}(w)$. Let $w_{0}=\argmin_{w}{L}_{\mu}(w)$ and $w'_{0}=\argmin_{w}{L}_{\mu'}(w)$. A similar argument as above gives
\begin{align}
L_{\mu}(w_{0})\le L_{\mu}(w_{0})-L_{\mu'}(w'_{0})+L_{\mu'}(w'_{0})\le L_{\mu}(w'_{0})-L_{\mu'}(w'_{0})+L_{\mu'}(w'_{0})\le \sqrt{2\sigma^{2}D(\mu'\|\mu)}+L_{\mu'}(w'_{0}).\label{eqnttt}
\end{align}
Using \eqref{eq600}, we deduce
\begin{align}
   L_{\mu'}(W(S')) \le \min_{w\in\mathcal{W}}L_{\mu'}(w)+\epsilon(\A,  \mathcal{P}, n, \delta)+ \sum_{i=1}^{n}\beta_{i}+2\sqrt{2\sigma^{2}\gamma}.
\end{align}
This completes the proof for the first upper bound.

For the second upper bound \eqref{eq101}, the population risk of the learning algorithm with respect to $\mu'$ by using Lemma \ref{lemma1} could be written as,
\begin{align}
	\lambda{L}_{\mu'}(\A(S'))&=\int_{\mathcal{Z}}\lambda\ell(\A(S'),z)\mu'(dz)\nonumber\\
	&=\mathbb{E}_{Z\sim \mu'}\left[\lambda\ell(\A(S'),Z)\right]\nonumber\\
	&\le\ln\Big(\mathbb{E}_{Z\sim\mu}\left[\exp\big(\lambda\ell(\A(S'),Z)\big)\right]\Big)+ D(\mu'\|\mu).\label{eq112}
\end{align}	

Note that both sides of \eqref{eq112} are random variables (and functions of $S'$) and $Z$ is taken to be independent of $S'$.

Considering the stability notion of algorithm from Definition \ref{def0}, the following inequality holds almost surely:
\begin{align*}
	\Big|\ell(\A(S'),z)-\ell(\A(S),z)\Big|
	\le \sum_{i=1}^{n}\beta_{i}(n)\boldsymbol{1}[Z_{i}\neq Z'_{i}], \qquad \forall z\in\mathcal{Z}. 
\end{align*}
Therefore, if we take $Z\sim\mu$ independent of $(S,S')$, we deduce that
\begin{align}
	\label{eq113}
	\Big|\ell(\A(S'),Z)-\ell(\A(S),Z)\Big|
	\le \sum_{i=1}^{n}\beta_{i}(n)\boldsymbol{1}[Z_{i}\neq Z'_{i}].
\end{align}

Next, we bound the random variable $\mathbb{E}_{Z\sim\mu}\left[\exp\big(\lambda\ell(\A(S'),Z)\right]$ in \eqref{eq112} from above as follows:
\begin{align}
	\mathbb{E}_{Z\sim\mu}\left[\exp\big(\lambda\ell(\A(S'),Z)\big)\right]&\overset{(a)}{\le} \mathbb{E}_{Z\sim\mu}\left[\exp\left(\lambda\ell(\A(S),Z)+\lambda\sum_{i=1}^{n}\beta_{i}\boldsymbol{1}[Z_{i}\neq Z'_{i}]\right)\right]\nonumber\\
	&\overset{(b)}{\le}\exp\Big(\lambda\mathbb{E}_{Z}\ell(\A(S),Z)+\lambda^{2}\sigma^{2}/2+\lambda\sum_{i=1}^{n}\beta_{i}\boldsymbol{1}[Z_{i}\neq Z'_{i}]\Big)\nonumber\\
	&\overset{(c)}{\le}_{1-\delta}\exp\Big(\lambda\min_{w}L_{\mu}(w)+\lambda\epsilon+\lambda^{2}\sigma^{2}/2+\lambda\sum_{i=1}^{n}\beta_{i}\boldsymbol{1}[Z_{i}\neq Z'_{i}]\Big),\label{eq1126}
	\end{align}
	where $(a)$ comes from \eqref{eq113}, inequality $(b)$ comes from the subgaussianity of $\ell(\A(S),Z)$ in terms of $Z$ for any fixed $S$ and  $(c)$ is derived from the definition of $\epsilon=\epsilon(\A,  \mathcal{P}, n, \delta)$. 
	
Next, from Markov's inequality we have
\begin{align}
&\exp\Big(\lambda\min_{w}L_{\mu}(w)+\lambda\epsilon+\lambda^{2}\sigma^{2}/2+\lambda\sum_{i=1}^{n}\beta_{i}\boldsymbol{1}[Z_{i}\neq Z'_{i}]\Big)\nonumber
\\
		&\overset{(d)}{\le}_{1-\delta'}\frac{1}{\delta'}\mathbb{E}_{S,S'}\exp\left(\lambda\min_{w}L_{\mu}(w)+\lambda\epsilon+\lambda^{2}\sigma^{2}/2+\lambda \sum_{i=1}^{n}\beta_{i}\boldsymbol{1}[Z_{i}\neq Z'_{i}]\right)\nonumber\\
	&=\frac{1}{\delta'}\exp(\lambda^{2}\sigma^{2}/2)\exp(\lambda\min_{w}\mathsf{L}_{\mu}(w)+\lambda\epsilon)\nonumber\\
	&\qquad\times\prod_{i=1}^{n}\left(\exp(\lambda\beta_{i})\cdot\|\mu-\mu'\|_{TV}+1-\|\mu-\mu'\|_{TV}\right),\label{eqnlast}
\end{align}
where the last equality follows from \eqref{coupling}.

Using \eqref{eq112}, \eqref{eq1126} and \eqref{eqnlast}, we find the following upper bound on ${L}_{\mu'}(\A(S'))$ with the probability at least $1-\delta-\delta'$:
\begin{align}\label{eq1112}
	{L}_{\mu'}(\A(S'))&\le\min_{w}{L}_{\mu}(w)+\epsilon+\frac{1}{\lambda}\log(1/\delta')+\lambda\sigma^{2}/2
	\\&\qquad+\frac{1}{\lambda}\sum_{i=1}^{n}\ln\big(1+(\exp(\lambda\beta_{i})-1)\|\mu-\mu'\|_{TV}\big)+\frac{1}{\lambda}{D}(\mu'\|\mu)\nonumber\\
	&\leq \min_{w}{L}_{\mu}(w)+\epsilon+\frac{1}{\lambda}\log(1/\delta')+\lambda\sigma^{2}/2+\frac{1}{\lambda}\sum_{i=1}^{n}\ln\big(1+\frac{1}{2}(\exp(\lambda \beta_{i})-1)\sqrt{\gamma}\big)+\frac{1}{\lambda}\gamma\nonumber
\end{align}
where we used  $D(\mu'\|\mu)\le \gamma$ and Pinsker's inequality. The choice of $\lambda=g(\delta')=\sqrt{\frac{2\left[\log(1/\delta')+\gamma\right]}{\sigma^{2}}}$ yields
\begin{align*}
	{L}_{\mu'}(\A(S'))&\le_{1-\delta-\delta'}\min_{w}{L}_{\mu}(w)+\epsilon(\A,  \mathcal{P}, n, \delta)\\&\qquad+\sqrt{2\sigma^{2}\left[\log(1/\delta')+\gamma\right]}+\frac{1}{g(\delta')}\sum_{i=1}^{n}\ln\left(1+\frac{1}{2}(\exp(g(\delta')\beta_{i})-1)\sqrt{\gamma}\right).
\end{align*}
Therefore, we get following upper bound for ${L}_{\mu'}(\A(S'))$ with probability  at least $1-\delta$:
\begin{align*}
	{L}_{\mu'}(\A(S'))&\le\min_{w}{L}_{\mu}(w)+\epsilon(\A,  \mathcal{P}, n, \delta/2)+\sqrt{2\sigma^{2}\left[\log(2/\delta)+\gamma\right]}\\&\qquad+\frac{1}{g(\delta/2)}\sum_{i=1}^{n}\ln\left(1+\frac{1}{2}(\exp(g(\delta/2)\beta_{i})-1)\sqrt{\gamma}\right).
\end{align*}

Finally, using \eqref{eqnttt}, we get
\begin{align}
	{L}_{\mu'}(\A(S'))\nonumber\le \min_{w}{L}_{\mu'}(w)+\epsilon(\A,  \mathcal{P}, n, \delta/2)+f(\delta)\nonumber
\end{align}
where
\[
f(\delta)\triangleq\sqrt{2\sigma^{2}\gamma}+\sqrt{2\sigma^{2}\left[\log(2/\delta)+\gamma\right]}\nonumber+\frac{1}{g(\delta/2)}\sum_{i=1}^{n}\ln\left(1+\frac{1}{2}(\exp(g(\delta/2)\beta_{i})-1)\sqrt{\gamma}\right).
\]
This completes the proof.
\subsection{Proof of Theorem \ref{thm11}}\label{appenEE}
It is clear that $\tilde{\D}_{2}(r/n)\le \D_{2}(r/n)$ from their definitions.
By the definition of $v_n$ for any arbitrary $p_{W'|S'}$ where  $S'=(Z'_1, Z'_2,\cdots, Z'_n)$ we have
\[
\mathbb{E}\left[\sum_{i=1}^n\frac1n\tilde\ell(W',Z'_i)\right]\geq v_n.
\]
It follows that
\begin{align*}
\D_{1}(r)&=
 \max_{P_{W'|S'}:~I(W';S')\leq r}
\mathbb{E}\left[L_{\mu}(W')-L_{S'}(W')\right]
\\&
=\max_{\substack{P_{W'|S'}:~I(W';S')\leq r,\\ \mathbb{E}\left[\sum_{i=1}^n\frac1n\tilde\ell(W',Z'_i)\right]\geq v_n}}
\mathbb{E}\left[L_{\mu}(W')-L_{S'}(W')\right].
\end{align*}
A similar argument as in \eqref{eqnRR12} shows that for any arbitrary $p_{W'|S'}$ we have
\[I(W';S')\geq \sum_{i=1}^n I(W';Z'_i).\]
Thus,
\begin{align*}
\D_{1}(r)&\leq
\max_{\substack{P_{W'|S'}:~\frac1n\sum_{i=1}^n I(W';Z'_i)\leq \frac{r}{n},\\ \frac1n\sum_{i=1}^n\mathbb{E}\left[\tilde\ell(W',Z'_i)\right]\geq v_n}}
\frac{1}{n}\sum_{i=1}^{n}\mathbb{E}\left[L_{\mu}(W')-\ell(W',Z'_{i})\right].
\end{align*}
Take some arbitrary $p_{W'|S'}$ and a time-sharing random variable $Q$ uniform on $\{1,2,\cdots, n\}$, independent of previously defined variables. Note that
\begin{align}
I(W';Z'_Q)&\leq \nonumber
I(Q,W';Z'_Q)\\
&=I(W';Z'_Q|Q)\label{eqnpf11a}
\\&=\frac1n\sum_{i=1}^n I(W';Z'_i)\nonumber
\\&\leq \frac{r}{n}\nonumber
\end{align}
where \eqref{eqnpf11a} follows from the fact that $Z'_i$'s are iid. We also have
\[\mathbb{E}\left[\tilde\ell(W',Z'_Q)\right]=
\frac1n\sum_{i=1}^n\mathbb{E}\left[\tilde\ell(W',Z'_i)\right]\geq v_n,
\]
Thus, the joint distribution $p_{W',Z'_Q}$ satisfies the constraints of $\tilde{\D}_2(\frac{r}{n})$. Moreover, $Z'_Q\sim\mu'$ and
\[\mathbb{E}\left[L_{\mu}(W')-\ell(W',Z'_{Q})\right]=
\frac{1}{n}\mathbb{E}\sum_{i=1}^{n}\left[L_{\mu}(W')-\ell(W',Z'_{i})\right]
\]
Thus, we deduce that $\D_1(r)\leq \tilde{\D}_2(\frac{r}{n})$ as desired.

\bibliographystyle{IEEEtran}
\bibliography{reference}
\end{document}